\documentclass[pdflatex,sn-mathphys,Numbered]{sn-jnl}% Math and Physical Sciences Reference Style
\usepackage{graphicx}%
\usepackage{multirow}%
\usepackage{amsmath,amssymb,amsfonts}%
\usepackage{amsthm}%
\usepackage{mathrsfs}%
\usepackage[title]{appendix}%
\usepackage{xcolor}%
\usepackage{textcomp}%
\usepackage{manyfoot}%
\usepackage{booktabs}%
\usepackage{algorithm}%
\usepackage{algorithmicx}%
\usepackage{algpseudocode}%
\usepackage{listings}%
\usepackage{enumerate}

\usepackage{bm}
\usepackage{txfonts}
\usepackage{mathtools} 
\mathtoolsset{showonlyrefs}
\newtheorem{theorem}{Theorem}%  meant for continuous numbers
\newtheorem{corollary}{Corollary}[theorem]
\newtheorem{proposition}{Proposition}% 
\newtheorem{remark}{Remark}%

\usepackage[normalem]{ulem}

\newcommand{\eps}{\varepsilon}
\renewcommand{\phi}{\varphi}

\newcommand{\ftrunc}[1]{f_{({#1})}}
\newcommand{\ftruncE}[1]{\ftrunc{N}^\equilibrium}
\newcommand{\PP}{\mathcal{P}}
\newcommand{\R}{\varmathbb{R}}
\newcommand{\inta}{\int_{\R^3} \int_0^\infty}
\newcommand{\da}{\,dI\,d\cc}

\newcommand{\G}{\mathscr{G}}
\newcommand{\TT}{\mathsf{T}}
\newcommand{\pp}{\mathsf{P}}
\newcommand{\RR}{\mathsf{R}}

\newcommand{\Sett}{\mathscr{S}}
\newcommand{\xx}{\mathbf{x}}
\newcommand{\vv}{\mathbf{v}}
\newcommand{\cc}{\bm{\xi}}
\newcommand{\CC}{\mathbf{C}}

\newcommand{\dI}{\,dI}

\newcommand{\dcc}{\,d\cc}

\newcommand{\kB}{k_B}
\newcommand{\kBm}{\frac{\kB}{m}}
\newcommand{\ikBm}{\frac{m}{\kB}}

\newcommand{\TK}{T^K}
\newcommand{\TI}{T^I}
\newcommand{\TKrel}{\TK_{\mathrm{rel}}}
\newcommand{\TIrel}{\TI_{\mathrm{rel}}}

\newcommand{\eIrel}{{\varepsilon^I_{\mathrm{rel}}}}

\newcommand{\eKrel}{{\varepsilon^K_{\mathrm{rel}}}}

\newcommand{\epsK}{\eps^K}
\newcommand{\epsI}{\eps^I}
\newcommand{\epsKinv}{\eps^{K,\,-1}}
\newcommand{\epsIinv}{\eps^{I,\,-1}}

\newcommand{\GE}{\G^{(K)}}
\newcommand{\GI}{\G^{(I)}}

\newcommand{\tes}{\tau_\mathrm{ES}}

\newcommand{\hatx}{\hat{x}_1}
\newcommand{\hatcx}{\hat{\xi}_1}
\newcommand{\hatcy}{\hat{\xi}_2}
\newcommand{\hatcz}{\hat{\xi}_3}
\newcommand{\hatcr}{\hat{\xi}_r}
\newcommand{\hatcc}{\hat{\cc}}
\newcommand{\hatci}{\hat{\xi}_i}

\newcommand{\hatvx}{\hat{v}_1}
\newcommand{\hatvy}{\hat{v}_2}
\newcommand{\hatvz}{\hat{v}_3}
\newcommand{\hatvi}{\hat{v}_i}
\newcommand{\hatvv}{\hat{\vv}}
\newcommand{\hatvj}{\hat{v}_j}
\newcommand{\hatrho}{\hat{\rho}}
\newcommand{\hatp}{\hat{\pp}}
\newcommand{\hatpij}{\hat{\pp}_{ij}}
\newcommand{\hatT}{\hat{T}}
\newcommand{\hatTK}{\hat{T}^K}
\newcommand{\hatTI}{\hat{T}^I}
\newcommand{\hatTrel}{\hat{T}_{\mathrm{rel}}}
\newcommand{\hatTKrel}{\hat{T}^K_{\mathrm{rel}}}
\newcommand{\hatTIrel}{\hat{T}^I_{\mathrm{rel}}}
\newcommand{\hateKrel}{\hat{\varepsilon}^K_{\mathrm{rel}}}
\newcommand{\hateIrel}{\hat{\varepsilon}^I_{\mathrm{rel}}}
\newcommand{\hatTT}{\hat{\TT}}
\newcommand{\hatq}{\hat{q}}
\newcommand{\hatf}{\hat{f}}
\newcommand{\hatI}{\hat{I}}
\newcommand{\hatA}{\hat{A}}
\newcommand{\hatAc}{\hat{A}_c}
\newcommand{\hatG}{\hat\G}
\newcommand{\hatGE}{\hat\G^{(K)}}
\newcommand{\hatGI}{\hat\G^{(I)}}
\newcommand{\hatQ}{\hat{Q}}
\newcommand{\hateps}{\hat\eps}
\newcommand{\hatepsK}{{\hat\eps}^K}
\newcommand{\hatepsI}{{\hat\eps}^I}
\newcommand{\hatphiI}{\hat\phi\left(\hatI\right)}

\newcommand{\normrho}{\hat{\hat\rho}}
\newcommand{\normv}{\hat{\hat v}}
\newcommand{\normT}{\hat{\hat T}}
\newcommand{\normTtr}{\hat{\hat T}^K}
\newcommand{\normTint}{\hat{\hat T}^I}

\newcommand{\dhatI}{\,d\hatI}
\newcommand{\dhatcx}{\,d\hatcx}
\newcommand{\dhatcc}{\,d\hatcc}
\newcommand{\dhatcr}{\,d\hatcr}

\newcommand{\abs}[1]{\left|{#1}\right|}
\newcommand{\Natural}{\mathbb{N}}
\newcommand{\Rthree}{{\mathbb{R}^3}}

\newcommand{\phiI}{\varphi\left(I\right)}
\newcommand{\Lap}[3]{\mathcal{L}\left[#1\left(#2\right)\right]\left(#3\right)}
\newcommand{\invLap}[3]{\mathcal{L}^{-1}\left[#1\left(#2\right)\right]\left(#3\right)}

\renewcommand{\geq}{\geqslant}
\renewcommand{\leq}{\leqslant}

\begin{document}

\title[A Novel ES-BGK Model for Non-Polytropic Gases]%
{A Novel ES-BGK Model for Non-Polytropic Gases with Internal State Density Independent  of the Temperature}


\author[1,2]{\fnm{Takashi} \sur{Arima}}\email{arima@tomakomai-ct.ac.jp}
%\email{takashi.arima@unibo.it}

\author[2,3]{\fnm{Andrea} \sur{Mentrelli}}\email{andrea.mentrelli@unibo.it}

\author*[2]{\fnm{Tommaso} \sur{Ruggeri}}\email{tommaso.ruggeri@unibo.it}

%\affil[0]{\orgdiv{Department}, \orgname{Organization}, \orgaddress{\street{Street}, \city{City}, \postcode{610101}, \state{State}, \country{Country}}}

\affil[1]{\orgdiv{Department of Engineering for Innovation}, \orgname{National Institute of Technology, Tomakomai College}, \orgaddress{\city{Tomakomai}, \postcode{059-1275}, \country{Japan}}}

\affil[2]{\orgdiv{Department of Mathematics \& Alma Mater Research Center on Applied Mathematics (AM$^2$) }, \orgname{University of Bologna},\orgaddress{\street{Via Saragozza 8}, \city{Bologna}, \postcode{40123}, \country{Italy}}}

%\affil[3]{\orgdiv{Alma Mater Research Center on Applied Mathematics (AM$^2$)}, \orgname{University of Bologna}, \orgaddress{\street{Via Saragozza 8}, \city{Bologna}, \postcode{40123}, \country{Italy}}}

\affil[3]{\orgdiv{Istituto Nazionale di Fisica Nucleare (I.N.F.N.), Sezione di Bologna, I.S. FLAG}, \orgaddress{\street{Viale Berti Pichat 6/2}, \city{Bologna}, \postcode{40127}, \country{Italy}}}

\abstract{
A novel ES-BGK-based model of non-polytropic rarefied gases in the framework of kinetic theory is presented. Key features of this model are: an internal state density function depending only on the microscopic energy of internal modes (avoiding the dependence on temperature seen in previous reference studies); full compliance with the H-theorem; feasibility of the closure of the system of moment equations based on the maximum entropy principle, following the well-established procedure of Rational Extended Thermodynamics. 

The structure of planar shock waves in carbon dioxide (CO$_2$) obtained with the present model is in general good agreement with that of previous results, except for the computed internal temperature profile, which is qualitatively different with respect to the results obtained in previous studies, showing here a consistently monotonic behavior across the shock structure, rather than the non monotonic behavior previously found. 
}



\maketitle

\section{Introduction \label{sec:introduction}}

The kinetic description of the nonequilibrium flow of rarefied polyatomic gases has been given much attention in recent years \cite{McCourt-1990, Nagnibeda-2009, Boyd-2017, Borsoni-2022}, and its importance for various applications, such as atmospheric re-entry problems, is now recognized \cite{Li-2009, Mathiaud-2018}.

One possible extension of the kinetic theory of monatomic gases to polyatomic gases was made, for the case of polytropic fluids, by Borgnakke and Larsen \cite{Borgnakke-1975}. According to the model presented in \cite{Borgnakke-1975}, the distribution function $f\equiv f(t,\xx,\cc,I)$ depends, in addition to time $t$, the space variable $\xx$, and  the molecular velocity $\cc$,  on an additional continuous variable $I$ representing the microscopic energy of the internal modes of a molecule, accounting for the energy exchange (other than the translational one) due to rotational and vibrational molecular motions.  This model, initially developed for Monte Carlo simulations of polyatomic gases, was later applied to the derivation of the generalized Boltzmann equation by Bourgat et al. \cite{Bourgat-1994}.

In this model, along the energy variable $I$, the state density function $\phiI$ needs to be introduced when constructing the macroscopic fields as moments of the distribution function $f$ integrated over the phase space of the velocity and the newly introduced microscopic energy variable. Being $\phi$ a state density, $\phiI dI$ represents the number of internal states between $I$ and $I+dI$, and it is defined as recovering the macroscopic total specific internal energy $\eps$. Therefore, the quantity $f\left(t,\xx,\cc,I\right) \phiI d\xx\,d\cc dI$ represents the number of molecules in the 7-dimensional phase space around a point $\left(\xx, \cc, I\right)$ at time $t$ \footnote{
    It should be remarked that the distribution function adopted by other authors, as for example in \cite{Kosuge-2019}, which is written as $f_*$ here, is related to the distribution function $f$ of the present paper as follows: 
    \begin{align}
    	f_*\left(t,\xx,\cc,I\right)= m^2 f\left(t,\xx,\cc,I\right)\, \phiI.
    \end{align}
}.

The internal  energy for polyatomic gases  is given  by the sum of the specific translational energy, $\epsK$, and the specific internal energy due to rotational and vibrational modes, $\epsI$:
\begin{align} \label{Poly:EnergyDef}
	\begin{split}
	&\eps = \epsK + \epsI, \\
	&\epsK = \frac{1}{\rho} \int_{\Rthree} \int_{0}^{\infty} \frac{1}{2} m C^2 f \phiI \da, \\
	&\epsI = \frac{1}{\rho} \int_{\Rthree} \int_{0}^{\infty} I f  \phiI \da,
	\end{split}
\end{align} 
where $\CC = \cc - \vv$ is the relative (peculiar) velocity, $\rho$ is the mass density, $\vv$ is the macroscopic (bulk) velocity, and $m$ denotes the molecular mass. 
For polytropic gases the specific internal energy $\eps$ is a linear function of the temperature:
\begin{align} \label{internale}
	\eps = \frac{D}{2} \kBm T, 
\end{align}
and the state density function $\phiI$ has the following expression:
\begin{align} \label{phi-poly}
	\phiI = I^{\left(D-5\right)/2}, 
\end{align}
where the gas-specific constant $D \, (\geq 3)$ represents the degrees of freedom of a molecule, $\kB$ is the Boltzmann constant, and $T$ denotes the absolute temperature.

At the kinetic level, it is assumed that the distribution function $f$ satisfies the Boltzmann equation
\begin{equation} \label{eq:Boltzmann}
	\frac{\partial f}{\partial t} + \xi_i \frac{\partial f}{\partial x_i} = Q\left(f\right),
\end{equation}
which is formally the same as the Boltzmann equation for monatomic gases, but with a collision integral $Q\left(f\right)$ taking now into account the influence of the internal degrees of freedom through the collisional cross-section. This model was proven to satisfy the H-theorem \cite{Bourgat-1994}. 

At the macroscopic level, in the framework of Rational Extended Thermodynamics (RET) \cite{Ruggeri-2021}, the system of 14 moments associated with Eq.~\eqref{eq:Boltzmann} was closed by Pavi\'c et al. \cite{Pavic-2013} \footnote{In this paper there are some typos that was corrected in \cite{Ruggeri-2021-MEP} and Chapter 7 of \cite{Ruggeri-2021} considering the polytropic case as particular case of nonpolytropic one.}
making use of the Maximum Entropy Principle (MEP) \cite{Kogan-1969, Dreyer-1987, Muller-1993}, stating that the  distribution function is the one that maximizes the entropy density 
\begin{align} \label{hgen}
	h = \inta H\left(f\right)\, \phiI \da,
\end{align}
with
\begin{align}\label{HHH}
	H\left(f\right) = - \kB f \log f,
\end{align}
under the constraint of prescribed moments (see for a brief history of MEP the Appendix  \ref{App:Max}). It is proven the equivalence of this approach to the one in which the system of model equations is obtained by means of a phenomenological closure by Arima et al. \cite{Arima-2012}. In subsequent years, the theory was successfully applied to the study of wave propagation, such as shock wave propagation in polyatomic gases (see \cite{Ruggeri-2021} and reference therein).

\smallskip 

The extension of the kinetic model of polytropic gases to non-polytropic gases, for which the internal energy depends on the temperature in a non-linear fashion, was undertaken by various authors following significantly different approaches. 

Kosuge et al. \cite{Kosuge-2019} proposed to replace, in Eq.~\eqref{internale} and Eq.~\eqref{phi-poly}, the constant $D$ with a temperature-dependent function, $D\left(T\right)$, allowing to model any arbitrary nonlinear dependence on the temperature of the internal energy $\eps\left(T\right)$ (a brief review of this reference model will be outlined in Sect.~\ref{sec:state-density}). This idea has the advantage of being simple, but it has two major weak points: Firstly, the resulting model equations with a model of the collisional term, which is discussed later, do not fulfill the H-theorem, as the authors themselves point out \cite{Kosuge-2019}; 
second, in the framework of this model, it is not possible to construct a closure of the moment equations in the spirit of RET by means of the usual procedure of MEP. This is because $\varphi\left(I, T\right)$ is now a function not only of the microscopic energy $I$ but also of the temperature $T$, which is, of course, a macroscopic field variable and therefore a moment of the distribution function itself. %This fact will be pointed out in Remark \ref{rem:MEP}.

In addition to that, the quantity $\phiI dI$ loses its neat physical meaning, since it does not represent anymore the number of internal states between $I$ and $I+dI$.

A different approach was proposed by Ruggeri and collaborators \cite{Ruggeri-2021, Arima-2016, Bisi-2018}, who noticed that $\phiI$, which should not depend on any field variables, is actually the inverse Laplace transform of a quantity that is related to the caloric equation of state of the internal modes, therefore leading to a state density depending only on $I$, but different from the one given in Eq.~\eqref{phi-poly}. 
In the framework of this model, the system of moment equations can be closed by means of the MEP, as well-established in RET, and field equations are indeed derived for non-polytropic gases in particular cases \cite{Ruggeri-2021,Arima-2021} including the 14 moment case \cite{Ruggeri-2021-MEP}. It is worth noticing that, in general, the procedure of the Laplace inversion required by this approach has to be carried out numerically, except for simple cases for which the Laplace inversion can be done analytically. However, it is also worth noticing that -- as it will be pointed out in Remark \ref{rem1} -- the Laplace inversion is actually not required explicitly as long as the field equations of macroscopic variables are needed \cite{Ruggeri-2021,Ruggeri-2021-MEP,Arima-2021}.

\smallskip 

When we deal with the Boltzmann equation, another critical model assumption has to be made concerning the explicit form of the collisional term $Q\left(f\right)$. For polyatomic gases, several models of simplified collisional terms have been proposed. We mention, among the others, the extension of the Bhatnagar-Gross-Krook (BGK) model \cite{Baranger-2020, Bisi-2017, Bisi-2018, Rahimi-2014, Struchtrup-1999}, the ellipsoidal statistical BGK (ES-BGK) model \cite{Holway-1966, Andries-2000, Brull-2009}, and the Fokker-Planck model \cite{Brau-1967, Gorji-2013, Mathiaud-2017}, all of which were originally developed for monatomic gases.
Among the above-mentioned models, the BGK-type collision term is -- due to its simplicity -- one of the most appealing and used models, but it has the well-known drawback of inducing by construction a Prandtl number equal to 1. 
In order to avoid this inconvenience in non-polytropic gases, Kosuge et. al., in their previously mentioned research paper \cite{Kosuge-2019}, proposed a model based on the ES-BGK collision term which allows to induce the correct Prandtl number, and studied in the framework of kinetic theory the structure of standing plane shock waves characterized by a large bulk viscosity, such as carbon dioxide (CO$_2$).

A model in which the molecular internal processes are treated in a more detailed way by accounting separately for the rotational and vibrational modes has been proposed by Arima et. al. in \cite{Arima-2017, Arima-2018}. In this model, two separate internal microscopic energies, $I^R$ for the rotational mode and $I^V$ for the vibrational mode, are introduced. In this case, two internal state densities, $\phi\left(I^R\right)$ and $\psi\left(I^V\right)$, are accordingly introduced. To model such processes, a generalized BGK model with $3$ relaxation times that satisfies the H-theorem is proposed \cite{Arima-2017}.

In the context of the ES-BGK model, a similar extension has been done by Dauvois et al. \cite{Dauvois-2021} and Mathiaud et al. \cite{Mathiaud-2022}. In these models, the H-theorem is satisfied; however, in contrast to the general case considered in \cite{Arima-2017} some particular assumptions were made: the contribution of the vibrational mode is treated as in the non-polytropic gas case, while it is assumed that the rotational mode behaves as in a polytropic gas. Since in these models the microscopic vibrational energy is assumed, by construction, to have only discrete energy levels, the state density function does not come into play. 
Although these models with separate internal modes allow to investigate the role of the molecular internal modes, the assumption of the relaxation equations of energies is needed in the construction of the ES-BGK model. 

\smallskip 

While previous studies have contributed significantly to the kinetic theory of non-polytropic gases, the development of an ES-BGK model with microscopic continuous energy levels is a task that remains to be accomplished: this is indeed the aim of the present paper. Specifically, we present here an ES-BGK model based on the microscopic continuous energy levels, $I$,  compatible with a state density of non-polytropic gases,  $\phiI$, independent of the temperature as it should be. The proposed model is conceptually different from all the models proposed in the above-mentioned papers \cite{Kosuge-2019, Dauvois-2021, Mathiaud-2022}, and it is proven to satisfy the H-theorem. At this stage of development of this new model, in order to avoid the assumption of the relaxation equations of the macroscopic rotational and vibrational energies as in \cite{Dauvois-2021, Mathiaud-2022}, the microscopic rotational and vibrational modes are treated as a whole for simplicity. This feature of the model has the additional advantage of allowing an easy integration of the model with experimental data concerning the total internal energy of the non-polytropic gas.
However, this assumption will be eliminated in a forthcoming refinement of the model. 

\smallskip 

A comparison of the numerical results pertaining planar shock waves obtained by adopting the present model to those obtained by adopting the reference model by Kosuge et al. \cite{Kosuge-2019} has been performed. Specifically, it will be shown that the model presented in \cite{Kosuge-2019} predicts a non-monotonic profile of the internal temperature through planar shock wave structures, while the correspondent profile obtained by the newly developed model, under the same conditions, is monotonic. All other macroscopic quantities appear to be, in all the examined cases, in a very good agreement with results presented in \cite{Kosuge-2019}.

\smallskip 

The paper is organized as follows. After summarizing, in Section~\ref{sec:state-density}, the relation between the state density and the internal energy, we introduce in Section~\ref{sec:newmodel} the new ES-BGK model for non-polytropic gases. In Section \ref{sec:reduced} the reduced ES-BGK model -- useful for reducing the computational cost of the numerical implementation of the model -- is introduced. Based on the reduced model, in Section~\ref{sec:numerical} we show the comparison of two ES-BGK models when the profiles of plane shock wave structures are computed. In Section~\ref{sec:conclusions}, concluding remarks will be outlined.

\section{Internal state density function \label{sec:state-density}}
Introducing the mass density $\rho$, the momentum density $\rho v_i$, and the energy density $\rho v^2/2 + \rho\eps$ as the first five moments of $f$:
\begin{equation} %\label{Poly:HydroMoments}
	\left(%
	\begin{array}{c}
		\rho \\
		\rho v_i \\
		\frac{\rho v^2}{2}  +  \rho \eps \\
	\end{array}%
	\right)
	= 
	\int_{\Rthree}  \int_{0}^{\infty}
	\left(
	\begin{array}{c}
		m \\
		m \xi_i \\
		\frac{m \xi^2 }{2}+   I  \\
	\end{array}
	\right)
	f\left(t, \xx, \cc, I\right)\, \phiI \dI \dcc,
\end{equation}
then from Eq.~\eqref{eq:Boltzmann}, taking into account the existence of the collision invariants
\begin{equation} %\label{Poly:Invariants}
	\left(m,\ m\xi_i,\ \frac{1}{2}m\xi^2 + I \right)^{T},
\end{equation}
we obtain the conservation laws of mass, momentum, and energy.

The total (specific) internal energy, 
\begin{equation}\label{energiatotale}
    \eps = \frac{1}{\rho} \int_{\Rthree} \int_{0}^{\infty} \left(\frac{1}{2} m C^2 +I\right) f \phiI \da =  \frac{1}{\rho} \int_{\Rthree} \int_{0}^{\infty} \left(\frac{1}{2} m C^2 +I\right) f^{(E)} \phiI \da = \eps_E,
\end{equation}
is an equilibrium quantity, while the (specific) translational energy, $\epsK$, and the (specific) internal mode energy, $\epsI$, defined in Eq.~\eqref{Poly:EnergyDef} are non-equilibrium ones:
\begin{equation}\label{temperatures}
	\eps = \epsK + \epsI =\epsK_E + \epsI_E,
\end{equation}
where $\epsK_E$ and $\epsI_E$ are, respectively, the equilibrium specific translational energy and the specific internal mode energy defined by
\begin{gather} 
    \epsK_E =\frac{1}{\rho} \int_{\Rthree} \int_{0}^{\infty} \frac{1}{2} m C^2 f^{(E)} \phiI \da, \label{Poly:EnergyDefEK} \\
	\epsI_E = \frac{1}{\rho}\int_{\Rthree} \int_{0}^{\infty} I f^{(E)} \phiI \da, \label{Poly:EnergyDefEI}
\end{gather}
$f^{(E)}$ being the equilibrium distribution function, which was obtained in \cite{Bourgat-1994} with considerations based on the H-theorem, and in \cite{Pavic-2013, Ruggeri-2021,Ruggeri-2021-MEP} requiring (similarly to the case of monatomic gas) the maximization of the entropy under the constraints of prescribed first five moments:
\begin{equation} \label{Poly:q(T)}
	f^{(E)} =  \frac{\rho}{m \, A\left(T\right) } \left( \frac{m}{2 \pi \kB T} \right)^{3/2} 
	\exp \left\{ -\frac{1}{\kB T} \left( \frac{1}{2} m C^2 + I \right) \right\} = f^{(M)}\, f^{(I)},
\end{equation}
where $f^{(M)}$ denotes the Maxwellian distribution function, and $f^{(I)}$ is the distribution function related to the internal mode:
\begin{equation} \label{fII}
	f^{(M)} = \frac{\rho}{m} \left( \frac{m}{2 \pi \kB T} \right)^{3/2} 
	\exp \left( -\frac{m C^2}{2 \kB T} \right), \qquad
	f^{(I)} = \frac{1}{A\left(T\right)}\exp \left( - \frac{I}{\kB T} \right),
\end{equation}
being
\begin{equation} \label{eq:A-norm}
	A\left(T\right)  = \int_{0}^{\infty} \exp \left( - \frac{I}{\kB T} \right) \phiI \dI
\end{equation}
a normalization factor such that
\begin{equation} \label{NfII}
	\int_0^\infty f^{(I)}\, \phiI \dI = 1.
\end{equation}
The function $A\left(T\right)$ can therefore be regarded, using the language of statistical mechanics, as the partition function for the molecular internal mode.

\smallskip 

For a rarefied non-polytropic gas, the total internal energy $\eps$ is a non-linear function of the temperature, the expression of which is given by the caloric equation of state\footnote{The temperature $T$ at kinetic level  in the non-polytropic gas can be defined as the inverse function of \eqref{caloric} with $\eps_E$ given by Eq.~\eqref{energiatotale}.}:
\begin{align} \label{caloric}
	\eps = \eps_E \left(T\right).
\end{align}
Once the specific heat $c_v\left(T\right) = d\eps/dT$ is known as a function of the temperature $T$, either as a result of statistical mechanics calculations, or by experimental measurements, the total internal energy $\eps$ is obtained as a function of the temperature $T$ by:
\begin{equation}% \label{eq:caloric}
	\eps\left(T\right) =  \int_{T_*}^{T} c_v\left(\tau\right)\, d\tau,
\end{equation}
where $T_*$ is a reference temperature. 
From Eq.~\eqref{Poly:EnergyDefEI}, Eq.~\eqref{Poly:q(T)} and Eq.~\eqref{eq:A-norm}, it is found (see \cite{Ruggeri-2021} and references therein):
\begin{equation} \label{eq:epsK+epsI}
	\epsI_E=\epsI_E\left(T\right) = \frac{1}{m} \int_{0}^{\infty} I f^{(I)}\, \phiI \dI = \kBm T^2 \frac{d \log A\left(T\right)}{dT}.
\end{equation}
Since it is known that
\begin{align}
	\epsK_E = \epsK_E\left(T\right) = \frac{3}{2}\kBm T, \label{eq:epsK}
\end{align}
if the caloric equation of state is given, the expression of the internal energy from Eq.~\eqref{temperatures} is obtained as
\begin{equation} \label{eq:epsI}
    \epsI_E\left(T\right) = \eps\left(T\right) - \epsK_E\left(T\right),
\end{equation}
and, from Eq.~\eqref{eq:epsK+epsI}, it is found
\begin{equation} \label{eq:AT}
	A\left(T\right) = A_0 \exp\left( \ikBm \int_{T_*}^{T} \frac{\epsI_E\left(\tau\right)}{\tau^2} \,d\tau\right),
\end{equation}
where $A_0$ is an inessential constant. Letting $s = 1 / \left(\kB T\right)$ and $\epsI_{E,s}\left(s\right) = \epsI_E\left(\frac{1}{\kB s}\right)$, Eq.~\eqref{eq:AT} can be written as
\begin{equation} \label{eq:As}
	A_s\left(s\right) = A\left(\frac{1}{\kB s}\right) = A_0 \exp\left( - \int_{1/\kB s^*}^{1/\kB s} m \epsI_{E,s}\left(\sigma\right) \,d\sigma\right).
\end{equation}
On the other hand, according to Eq.~\eqref{eq:A-norm} the  function $A_s$ is
\begin{equation} %\label{eq:As-norm}
	A_s\left(s\right) = \int_{0}^{\infty} e^{-s I} \phiI\,dI,
\end{equation}
from which it is seen that the function $A_s$ is the Laplace transform of $\phiI$ \cite{Arima-2016, Bisi-2018, Ruggeri-2021}:
\begin{equation} %\label{eq:As-2}
    A_s\left(s\right) = \Lap{\phi}{I}{s}.
\end{equation}
The internal state density function, $\phiI$, is therefore obtained as the inverse Laplace transform of the the function $A_s$ defined in Eq.~\eqref{eq:As}:
\begin{equation} \label{eq:invLap}
	\phiI = \invLap{A_s}{s}{I}.
\end{equation}

The inverse Laplace transform prescribed in Eq.~\eqref{eq:invLap} can be carried out analytically in simple cases, such as the case of a gas with constant specific heat $c_v$ (i.e. a polytropic gas), or the case of a gas with a specific heat $c_v$ which is a linear function of the temperature, which we show below. 

\smallskip

\begin{remark}\label{rem1}
Except for the cases of a gas with constant specific heat or linearly varying specific heat, in general (and realistic) cases of a gas with a specific heat which is a generic function of the temperature, the inverse Laplace transform prescribed by Eq.~\eqref{eq:invLap} is difficult (if even possible) to perform analytically, and we can perform it  only numerically.
On the other hand, it is remarkable that, in order to close -- making use of MEP -- the system obtained by taking moments {of the Boltzmann equation}, the explicit expression of $\phiI$ is actually not needed. In fact, it is proven that all coefficients in the constitutive equations are expressed by the integral of the equilibrium distribution function and, as a consequence, only the following type of integral appear:
\begin{align} \label{barI}
	\bar{A}_r = \int_0^{\infty} f^{(I)} \left(\frac{I}{\kB T}\right)^r \phiI dI, \qquad
    r \in \Natural,
\end{align}
which is a generalization of the moments appearing in Eq.~\eqref{NfII} and Eq.~\eqref{eq:epsK+epsI}. By differentiating Eq.~\eqref{eq:epsK+epsI} with respect to $T$, it is possible to find a recurrence formula such that the integrals $\bar{A}_r$ are determined for any $r \in \Natural$ by $\epsI_E\left(T\right)$ and its derivatives \cite{Arima-2021}. See also \cite{Ruggeri-2021-MEP, Ruggeri-2021} for particular cases.
\end{remark}

\smallskip

\begin{remark}
It is seen that the physical dimension of $A\left(T\right)/\phiI$ is the same as that of $I$ -- as it can be deduced from Eq.~\eqref{fII}$_2$, Eq.~\eqref{eq:A-norm} and Eq.~\eqref{NfII} -- which in turn corresponds to the dimension of $\kB T$. Furthermore, we notice that the physical dimension of $A\left(T\right)$ hinges on an inessential constant $A_0$, as shown in Eq.~\eqref{eq:AT}. In the case of a polytropic gas, which is discussed in the following, this physical dimension is deduced from Eq.~\eqref{eq:A0+As}.
\end{remark}

\subsection{Constant specific heat (polytropic gas)}

For a polytropic gas, the specific heat $c_v$ is constant, and it is expressed in terms of the molecular degrees of freedom $D$ as follows:
\begin{equation}
	\frac{c_v}{\kB / m} = \frac{D}{2}.
\end{equation}
As shown in Eq.~\eqref{internale}, the total internal energy $\eps$ is a linear function of the temperature; the internal energy due to the translational motion, $\epsK_E$, and the internal energy related to the internal degrees of freedom, $\epsI_E$, are given, respectively, by:
\begin{equation}
	\epsK_E\left(T\right) = \frac{3}{2} \kBm T, \qquad
	\epsI_E\left(T\right) = \frac{D-3}{2} \kBm T = \left(1+\alpha\right) \kBm T, 
\end{equation}
where
\begin{align}
	\alpha= \frac{D-5}{2}, \qquad \left(\alpha \geq -1\right),
\end{align}
or, 
\begin{equation}\label{Dcost}
	D= 5+ 2\alpha.
\end{equation}
From Eq.~\eqref{eq:AT} it is readily seen that
\begin{equation} %\label{eq:AT_cv_const-aht}
	A\left(T\right) = 
	%A_0 \exp\left( \int_{T_*}^{T} \frac{\left(1 + \alpha_0\right) \tau}{\tau^2} \,d\tau \right) =
	A_0 \exp\left( \int_{T_*}^{T} \frac{1 + \alpha}{\tau} \,d\tau \right) =
	A_0 \left(\frac{T}{T_*}\right)^{1 + \alpha},
\end{equation}
and 
\begin{equation}
	A_s\left(s\right) = A_0  \left(\frac{s_*}{s}\right)^{1 + \alpha}.
\end{equation}
From Eq.~\eqref{eq:invLap} it is obtained:
\begin{equation}
	\phiI = A_0  \frac{ I^{\alpha}}{\left({\kB T_*}\right)^{1 + \alpha}\Gamma\left(1 + \alpha\right)},
\end{equation}
and, letting,
\begin{equation} \label{eq:A0+As}
	A_0 = \left(\kB T_*\right)^{1 + \alpha}\Gamma\left(1 + \alpha\right), 
\end{equation}
it is found:
\begin{equation} \label{eq:A+phi-polytropic}
	\phiI = I^{\alpha}, \qquad A\left(T\right) =  \left(\kB T\right)^{1 + \alpha}\Gamma\left(1 + \alpha\right) ,
\end{equation}
which is compatible with Eq.~\eqref{phi-poly}.

\subsection{Linearly varying specific heat}

In the case of a specific heat, $c_v$, linearly depending on the temperature,
\begin{equation}
	\frac{c_v\left(T\right)}{\kB / m} =
	\frac{5}{2} + \alpha_0 + 2 \alpha_1\frac{T}{T_*}, 
\end{equation}
where $\alpha_0$ and $\alpha_1$ are dimensionless constants, on the basis of Eq.~\eqref{internale} -- which is valid only for polytropic gases -- we may write:
\begin{equation}
	\eps\left(T\right) = \frac{D\left(T\right)}{2} \kBm T, 
\end{equation}
where $D$, in contrast to Eq.~\eqref{Dcost}, is now a function of the temperature $T$:
\begin{equation}\label{dttt}
	D\left(T\right)=5 +2 \alpha_0  + 2 \alpha_1 \frac{T}{T_*}.
\end{equation}
The energy of the internal modes can now be written as:
\begin{equation}\label{lineare}
	\epsI_E\left(T\right) = \frac{D\left(T\right)-3}{2} \kBm T =\left(1 + \alpha_0  + \ \alpha_1 \frac{T}{T_*}\right) \kBm T,
\end{equation}
and, taking into account Eq.~\eqref{eq:AT}, and choosing $A_0$ as in Eq.~\eqref{eq:A0+As}, we obtain:
\begin{equation} \label{eq:AT2-o1}
    \begin{split}
    A\left(T\right)
	&= A_0 \exp\left\{ \int_{T_*}^{T} \left( \frac{1 + \alpha_0}{\tau} + \frac{\alpha_1}{T_*} \right) \,d\tau \right\} \\
	&= \left(\kB T\right)^{1 + \alpha_0} \Gamma\left(1 + \alpha_0\right) \, \exp \left\{\alpha_1\left(\frac{T}{T_*}-1\right)\right\},
    \end{split} 
\end{equation} 
and
\begin{equation} \label{eq:As2-o1}
	A_s\left(s\right) 
	= \exp\left(-\alpha_1\right) \Gamma\left(1 + \alpha_0 \right) \, \left(\frac{1}{s}\right)^{1 + \alpha_0}  \exp\left(\alpha_1 \frac{s_*}{s} \right).
\end{equation}
It can be proven that Eq.~\eqref{eq:As2-o1} has the following exact inverse Laplace transform: 
\begin{equation}\label{phi-gen}
	\phiI = \exp\left(-\alpha_1\right)\Gamma\left(1+\alpha_0\right) \, I^{\alpha_0} \left(\sqrt{\frac{\alpha_1 I}{\kB T_*}}\right)^{-\alpha_0}  \mathcal{I}_{\alpha_0}\left(2 \sqrt{\frac{\alpha_1 I}{\kB T_*}}\right),
\end{equation}
being $\mathcal{I}_{\alpha_0}\left(z\right)$ the modified Bessel function of the first kind or order $\alpha_0$. 
It can also be proven that
\begin{equation}
	\phiI \xrightarrow[]{\alpha_1 \to 0} I^{\alpha_0},
\end{equation}
and the state function $\phiI$ for a polytropic gas given in Eq.~\eqref{eq:A+phi-polytropic} is recovered as expected.

\smallskip 

As discussed in Sect.~\ref{sec:introduction}, in the model presented in \cite{Kosuge-2019}, the state density function -- based on Eq.~\eqref{lineare} -- would be defined for a gas with a linearly varying specific heat, as
\begin{align}\label{phi-Aoki}
	\phi\left(I, T\right) = I^{\left(D\left(T\right)-5\right)/2}, 
\end{align}
with $D\left(T\right)$ given in Eq.~\eqref{dttt}.
It is clearly seen that, despite corresponding to the same internal energy, the state density function $\phi$ given in Eq.~\eqref{phi-gen} is independent of the temperature $T$, while the state density function $\phi$ given in Eq.~\eqref{phi-Aoki} depends on the temperature $T$.

\section{Novel ES-BGK model for non-polytropic gas \label{sec:newmodel}}

In this Section, our novel ES-BGK model for non-polytropic polyatomic gases with temperature-dependent specific heat is described. In this model, the state density function $\phiI$ is not assumed to be given by Eq.~\eqref{eq:A+phi-polytropic}$_1$, which is valid only for polytropic gases; rather, $\phiI$ is obtained as the inverse Laplace transform of the function $A_s(s)$ given in Eq.~\eqref{eq:invLap}, in a fully consistent way.

\subsection{Nonequilibrium temperatures}
Before discussing the novel ES-BGK model, we introduce the nonequilibrium temperatures, $\TK$ and $\TI$, associated, respectively, to the molecular translational and internal modes. The temperatures $\TK$ and $\TI$ are implicitly defined by the internal energies of each mode, Eq.~\eqref{Poly:EnergyDef},  via the caloric equations of state of each mode:
\begin{align}
    \epsK = \epsK_E\left(\TK\right), \qquad
    \epsI = \epsI_E\left(\TI\right),
\end{align}
i.e.
\begin{align} \label{def:noneqT}
    \TK = \epsKinv_E\left(\epsK\right)=\frac{2\epsK}{3 \kBm}, \qquad 
    \TI = \epsIinv_E \left(\epsI\right),
\end{align}
$\epsKinv_E$ and $\epsIinv_E$ being the inverse functions of, respectively, $\epsK_E$ and $\epsI_E$, given in Eq.~\eqref{eq:epsK} and Eq.~\eqref{eq:epsI}. We remark that the translational temperature $\TK$ is related to the stress tensor 
\begin{align}
    t_{ij} =  - \inta m C_i C_j f \phiI \da,
    \label{eq:stress}
\end{align}
which is decomposed as follows 
\begin{align}\label{stress}
	t_{ij} = - \PP \delta_{ij} + \sigma_{\langle ij\rangle},
\end{align}
where $\sigma_{\langle ij\rangle}$ is the shear stress\footnote{Angular brackets denote the symmetric traceless part (deviatoric part) with respect to these indices.} and $\PP$ is the total pressure, the latter being the sum of the equilibrium pressure $p$, expressed as
\begin{align}
    p \equiv p\left(\rho, T\right) = \frac{2\rho}{3}\epsK_E\left(T\right) = \kBm \rho T,
\end{align}
and the dynamic pressure (nonequilibrium part of pressure) $\Pi = \PP - p$. From Eq.~\eqref{Poly:EnergyDef}$_2$ and Eq.~\eqref{eq:stress}, it is seen that the nonequilibrium energy of the translational mode is expressed in terms of the total pressure $\PP$ as follows:
\begin{align}\label{eK-P}
	\epsK=\epsK_E\left(\TK\right) = - \frac{1}{2\rho} t_{ll} = \frac{3}{2\rho} \PP. 
    % = \frac{1}{\rho} \inta \frac{m}{2}C^2 f \phiI \da;
\end{align}
Recalling the functional form of $\epsK_E$, given in Eq.~\eqref{eq:epsK}, we have
\begin{align}
    \PP = \kBm \rho \TK = p\left(\rho, \TK\right),
\end{align}
which, together with Eq.~\eqref{stress}, shows the relation between the stress tensor $t_{ij}$ and the translational temperature $\TK$.

\subsection{Model of collisional term}

The newly proposed ES-BGK model for non-polytropic polyatomic gases is the natural extension of the original model studied in \cite{Andries-2000}. The collision integral $Q$ is given by 
\begin{align} \label{eq:Q-ESBGK}
	Q\left(f\right) = \frac{1}{\tes}\left(\G - f\right),
\end{align}
where the relaxation time $\tes$ is a positive function of $\rho$ and $T$, and the distribution function $\G$ is determined as follows. 

Let $\Sett$ be the set of all non-negative, integrable distribution functions such that for any $\bar{\G} \in \Sett $, the following relations hold:
\begin{align}
    \begin{split}
        & \rho =  \inta m f \phiI \da = \inta m \, \bar{\G} \phiI \da ,\\
        & \rho v_i = \inta m \xi_i f \phiI \da = \inta m \xi_i \, \bar{\G} \phiI \da, \\
        & \rho \eps =  \inta  \left(\frac{m}{2}C^2+{I}\right) f \phiI \da = \inta \left(\frac{m}{2}C^2+{I}\right) \, \bar{\G} \phiI \da,
    \end{split}
    \label{momG5}
\end{align}
i.e. the first five moments of $\bar\G$ are equal to the corresponding moments of $f$.
It is important to mention that defining the set $\Sett$ as the set of the distribution functions $\bar\G$ for which Eq.~\eqref{momG5} holds, guarantees that the conservation laws are satisfied. In fact, multiplying the Boltzmann equation \eqref{eq:Boltzmann} with Eq.~\eqref{eq:Q-ESBGK} by each of the collision invariants $\left(m,\, m\xi_i,\, m\left(\xi^2+2\frac{I}{m}\right)\right)^T$ and integrating over the phase space with respect to molecular velocity and internal energy variable, the conservation laws of mass, momentum and energy are obtained:
\begin{align}
    \begin{split}
        &\frac{\partial \rho}{\partial t} + \frac{\partial }{\partial x_j}\left(\rho v_j\right) = 0 ,\\
		&\frac{\partial}{\partial t}\left(\rho v_i\right) + \frac{\partial}{\partial x_j}\left(\rho v_i v_j  - t_{ij} \right) = 0 ,\\
		&\frac{\partial}{\partial t} \left(\rho v^2 + 2\rho \eps\right)+ \frac{\partial}{\partial x_j} \left(\rho v^2 v_j + 2\rho \eps v_j - 2 t_{jk} v_k + 2q_j\right) = 0,
    \end{split}
    \label{conserve}
\end{align}
where $q_j$ is the heat flux defined by 
\begin{equation} \label{eq:tq} 
    q_j =  \inta  \frac{m}{2}\left(C^2 + 2\frac{I}{m}\right) C_j f \phiI \da.
\end{equation}
The distribution function $\G$ is determined by the following theorem.

\smallskip 

\begin{theorem}\label{theo:1}
Let us consider the following eleven moments of $\bar{\G} \in \Sett$:
\begin{align}\label{FG}
	\boldsymbol{F}^{\bar{\G}} = \begin{pmatrix}
		\rho \\ \rho v_i \\ \rho {\TT}_{ij} + \rho v_i v_j \\  2\rho \eIrel
	\end{pmatrix}
	 = m \inta \bm{\psi}\, \bar{\G} \phiI \da,
\end{align}
with $\bm{\psi} \equiv \left(1,\, \xi_i,\, \xi_i \xi_j,\, 2{I}/{m}\right)^T$, where we have introduced the second-order symmetric and positive definite tensor:
\begin{align}
    \label{momG11}
    \begin{split}
        \TT_{ij} = \frac{1}{\rho}\inta  m C_i C_j \, \bar{\G} \phiI \da,%\qquad \text{}.
        %\rho \eIrel = \inta I\, \G \phiI \da, 
        %\inta  \left(\frac{m}{2}C^2+{I}\right) \G \phiI \da = \rho \eps \qquad
        \end{split}
\end{align}
and the relaxation internal energy $\eIrel$:
\begin{equation}\label{momG11I}
    \eIrel = \frac{1}{\rho} \inta I\, \bar{\G} \phiI \da.
\end{equation}
Defining the entropy density in $\Sett $ as follows:
\begin{equation}\label{HG}
    h^{\bar{\G}} =  -\kB \inta \bar{\G} \log \bar{\G}  \phiI \da,
\end{equation}
the distribution function $\G \in \Sett $ that  maximizes the entropy \eqref{HG} under the constraint that the eleven moments of $\bar{\G}$ defined in Eq.~\eqref{FG} are prescribed, is 
\begin{equation} \label{def:G}
    \G = \GE\, \GI,
\end{equation}
with 
\begin{align} \label{eq:GE}
\begin{split}
    &\GE = \frac{\rho }{m\left(2\pi\right)^{3/2} \left(\det \TT\right)^{1/2} } \exp \left\{-\frac{1}{2} \left(\xi_i - v_i\right) \left(\TT^{-1}\right)_{ij} (\xi_j - v_j) \right\},  \\
    &\GI = \frac{1}{A\left(\TIrel\right)} \exp\left(-\frac{I}{\kB \TIrel}\right),
\end{split}
\end{align}
where $\TIrel$ is the relaxation temperature defined via the caloric equation of state given in Eq.~\eqref{eq:epsI}:
\begin{align} \label{def:TIrel2}
    \TIrel = \epsIinv_E\left(\eIrel\right). 
\end{align}
The entropy density given in Eq.~\eqref{HG} maximized by $\G$  has the following expression:
\begin{equation}\label{hGexpression}
    \begin{split} 
        h^\G 
        &= - \kBm \rho \left(\log \frac{\rho}{m\left(2\pi\right)^{3/2}\left(\det \TT\right)^{1/2} A\left(\TIrel\right)} - \frac{m\epsI_E\left(\TIrel\right)}{\kB \TIrel} - \frac{3}{2}\right).
    \end{split}
\end{equation}
\end{theorem}
\begin{proof}
MEP states that the distribution function $\G \in \Sett $ which maximizes the entropy density \eqref{HG} with prescribed eleven moments \eqref{FG} is the solution of a variational problem with constraints associated to the functional 
\begin{equation} \label{eq:LGbar}
	L\left(\bar{\G}\right) = -\kB \inta \bar{\G} \log \bar{\G}  \phiI \da + 
	\boldsymbol{\Lambda} \cdot \left(\boldsymbol{F}^{\bar{\G}} - m\inta \bm{\psi} \, \bar{\G} \phiI \da\right),
\end{equation}
where $\boldsymbol{\Lambda} \equiv \left(\lambda,\, \lambda_i,\, \lambda_{ij},\, \mu\right)$ is the vector of the  Lagrange multipliers. In order to have an extremum the first variation with respect to $\bar{\G}$ must be equal to zero, i.e.\footnote{We observe that the MEP cannot be done in the form   \eqref{eq:variation} in the case in which the state density $\varphi(I,T)$ also depends on the temperature $T$.}
\begin{equation}
   \frac{\delta L}{\delta \bar{\G}}  = - k_B \inta \left(\log \bar{\G} + 1 + \frac{m}{k_B}\boldsymbol{\Lambda} \cdot \bm{\psi}\right)  \phiI \da = 0,
    \label{eq:variation}
\end{equation}
and the distribution function $\G$ maximizing the functional \eqref{eq:LGbar} is \cite{Muller-1998, Ruggeri-2021, Boillat-1997}:
\begin{equation} \label{eq:fmax}
	\G = \exp\left(-1-\frac{m}{\kB}\chi\right), 
\end{equation}
where 
\begin{equation}\label{ident}
	\chi =\boldsymbol{\Lambda}\cdot \bm{\psi}
    =\hat{\boldsymbol{\Lambda}}\cdot \hat{\bm{\psi}},
\end{equation}
with a hat on a quantity indicating its velocity independent part: $\hat{\bm{\psi}} \equiv \left(1,\, C_i,\, C_i C_j,\, 2I/m\right)^T$ and $\hat{\boldsymbol{\Lambda}} \equiv \left(\hat{\lambda},\, \hat{\lambda}_i,\, \hat{\lambda}_{ij},\, \hat{\mu}\right)$. The identity in \eqref{ident} is proved in  \cite{Ruggeri-1989}, and it is evident also by the fact that $\G$ is a scalar independent of the frame.
For later convenience, we write $\G$ as
\begin{align}\label{eq:fmax2}
	\G = \Omega\, e^{-\tilde{\lambda}_i C_i} e^{- \tilde{\lambda}_{ij}C_i C_j - \frac{2\hat{\mu} }{\kB}I },
\end{align}
where 
\begin{align}
	\Omega = \exp \left(- 1 - \frac{m}{\kB}\hat{\lambda}\right), \qquad 
	\tilde{\lambda}_{i} = \frac{m}{\kB} \hat{\lambda}_i, \qquad 
	\tilde{\lambda}_{ij} = \frac{m}{\kB} \hat{\lambda}_{ij}.
\end{align}
Given that
\begin{equation} \label{eq:integrals}
    \begin{split}
	&\int_{\Rthree} e^{-\tilde{\lambda}_{ij}C_iC_j} d\CC = \pi^{3/2}\left(\det \tilde{\boldsymbol{\lambda}}\right)^{-\frac{1}{2}},\\
	&\int_{\Rthree} C_k C_l e^{-\tilde{\lambda}_{ij}C_iC_j} d\CC = \frac{\pi^{3/2}}{2}\left( \tilde{\boldsymbol{\lambda}}^{-1}\right)_{kl}\left(\det \tilde{\boldsymbol{\lambda}}\right)^{-\frac{1}{2}},
    \end{split}
\end{equation}
where $\tilde{\boldsymbol{\lambda}}$ is the matrix the elements of which are $\tilde{\lambda}_{ij}$, and inserting Eq.~\eqref{eq:fmax2} into the right-hand side of Eq.~\eqref{FG}, we obtain $\tilde{\lambda}_i=0$ and \footnote{Hereafter, $\TT_{ij}$ and $\eIrel$ are given in Eqs.~\eqref{momG11} and ~\eqref{momG11I} by substituting $\bar{\G}$ with  $\G$.}
\begin{align}
	&\rho = m \Omega \pi^{3/2}\left(\det \tilde{\boldsymbol{\lambda}}\right)^{-1/2}A\left(\frac{1}{2\hat{\mu}}\right),\\
	&\rho {\TT}_{ij} = \frac{1}{2}m \Omega \pi^{3/2}\left( \tilde{\boldsymbol{\lambda}}^{-1}\right)_{ij}\left(\det \tilde{\boldsymbol{\lambda}}\right)^{-1/2}A\left(\frac{1}{2\hat{\mu}}\right),\\
	&2\rho \eIrel = -2 \kB \Omega \pi^{3/2}\left(\det \tilde{\boldsymbol{\lambda}}\right)^{-1/2}\frac{ d A\left(\frac{1}{2\hat{\mu}}\right)}{d \left({2 \hat{\mu}}\right)},
\end{align}
and then
\begin{align}
	&\Omega = \frac{\rho}{m \pi^{3/2} \left(\det \tilde{\boldsymbol{\lambda}}\right)^{-1/2}A\left(\frac{1}{2\hat{\mu}}\right)}, \label{lambdas1} \\
	&\left( \tilde{\boldsymbol{\lambda}}^{-1}\right)_{ij} = 2 \TT_{ij}, \label{lambdas2} \\
	&\eIrel = - \kBm \frac{ d }{d \left({2 \hat{\mu}}\right)} \log A\left(\frac{1}{2\hat{\mu}}\right). \label{lambdas3}
\end{align}
From Eq.~\eqref{lambdas2}, we have $(\det \tilde{\boldsymbol{\lambda}})^{-1} = 2^3 \det {\TT}$ and, since it can be seen that $\hat{\mu}$ has the physical dimension of inverse temperature, we introduce a new temperature $\TIrel$ defined as
\begin{align}
	\TIrel = \frac{1}{2\hat{\mu}}.
\end{align} 
Recalling Eq.~\eqref{eq:epsK+epsI}, Eq.~\eqref{lambdas3} suggests:
\begin{align}
    \eIrel = \epsI_E\left(\TIrel\right);
\end{align}
in other words, $\TIrel$ is determined by $\eIrel$ from the inverse function of the caloric equation of the state of  internal mode, as introduced in Eq.~\eqref{def:TIrel2}.

The explicit expression of the entropy density maximized by $\G$, i.e. $h^\G$, given in Eq.~\eqref{hGexpression}, is derived by substituting $\G$ into Eq.~\eqref{HG}. 
\end{proof}

\medskip

\subsection{Derivation of $\TT_{ij}$ as a function of physical quantities}
Following the discussion in \cite{Brull-2009}, we find that the tensor $\TT_{ij}$ is related to the physical macroscopic variables. We can draw a parallel with the results of the standard BGK model. Specifically, for the collisional term given by Eq.~\eqref{eq:Q-ESBGK} and Eq.~\eqref{def:G},  we require that the following six relations hold:
\begin{equation} \label{prod11}
    \frac{m}{\tes} \inta
    \begin{pmatrix}
        \displaystyle
        \xi^2 \\ %[10pt] \displaystyle
        \xi_{\langle i} \xi_{j \rangle} %\\[10pt]
    \end{pmatrix}
    \left(\G-f\right) \ \phiI \da
    = 
    \begin{pmatrix}
        \displaystyle 
        -2\frac{\rho}{\tau}\left(\epsK_E\left(\TK\right)-\epsK_E\left(T\right)\right) \\ 
        \displaystyle
        \frac{1}{\tau_\sigma} \sigma_{\langle ij\rangle} 
    \end{pmatrix},
\end{equation}
where $\tau = \tau\left(\rho,T\right)$ and $\tau_\sigma = \tau_\sigma\left(\rho,T\right)$ are relaxation times. In the standard BGK model this is an identity but with a common relaxation time. 
In contrast, we now require that $\epsK$ and $\sigma_{\langle ij\rangle}$  have different relaxation times; in such a way, we can have a physically more appropriate Prandtl number when we take the hydrodynamic limit.

Although we will explore the hydrodynamic limit in detail in Sect.~\ref{sec:CE}, to clarify the meaning of the production terms in Eq.~\eqref{prod11}, we present the field equations for $\epsK$ and $\sigma_{\langle ij \rangle}$, obtained by multiplying the Boltzmann equation \eqref{eq:Boltzmann} by $\left(m\xi^2,\, m \xi_{\langle i}\xi_{j \rangle}\right)^T$ and integrating each of the two resulting equations over the phase space with respect to the molecular velocity and the internal energy variable:
\begin{align}
    \begin{split}
        &\frac{\partial}{\partial t} \left(\rho v^2 + 2 \rho \epsK\right)+ \frac{\partial}{\partial x_k}\left(\rho v^2 v_k + \frac{10}{3} \rho \epsK v_k - 2 \sigma_{\langle kl \rangle}v_l + \hat{H}^0_{llk} \right) = 
        - \frac{2\rho}{\tau}\left(\epsK_E\left(\TK\right) -\epsK_E\left(T\right)\right),\\
		&\frac{\partial}{\partial t}\left(\rho v_{\langle i}v_{j\rangle} - \sigma_{\langle ij\rangle}\right) + \frac{\partial}{\partial x_k}\left( \rho v_{\langle i}v_{j\rangle}v_k + \frac{4}{3}\rho \epsK v_{\langle i} \delta_{j\rangle k} - \sigma_{\langle ij\rangle}v_k - 2 \sigma_{\langle k \langle i \rangle}v_{j\rangle} + \hat{H}^0_{\langle ij\rangle k} \right) 
		= \frac{1}{\tau_\sigma}\sigma_{\langle ij\rangle},
    \end{split}
    \label{eksig}
\end{align}
where 
\begin{equation}
    \begin{gathered}
        \hat{H}^0_{llk} = \inta  {m}C^2 C_k f \phiI \da, \\
        \hat{H}^0_{\langle ij\rangle k} = \inta  {m}C_{\langle i}C_{j\rangle} C_k f \phiI \da.
    \end{gathered}
\end{equation}
Eq.~\eqref{temperatures} with Eq.~\eqref{conserve}$_{3}$ and Eq.~\eqref{eksig}$_1$ provide an equation for the relaxation of $\epsI$:
\begin{align}
	&\frac{\partial}{\partial t} \left( 2\rho \epsI\right)+ \frac{\partial}{\partial x_k} \left(2\rho \epsI v_k - \hat{H}^0_{llk} + 2q_k\right) = 
 -\frac{2\rho}{\tau}\left(\epsI_E\left(\TI\right) -\epsI_E\left(T\right)\right).
 \label{eqei}
\end{align}
From Eqs.~\eqref{eksig} it is seen that $\epsK$ and $\sigma_{\langle ij \rangle}$ relax to the equilibrium state with relaxation times $\tau$ and $\tau_\sigma$, respectively.  The role of the relaxation times is easily found when we consider the spatially homogeneous case. In this case, Eqs.~\eqref{eksig} and Eq.~\eqref{eqei} reduce to:
\begin{align} % \label{homogeneous}
\begin{split}
    &\frac{d{\eps}^K_E\left(\TK\right)}{dt} = - \frac{1}{\tau}\left(\epsK_E\left(\TK\right) -\epsK_E\left(T\right)\right),\\
    &\frac{d{\eps}^I_E\left(\TI\right)}{dt} = - \frac{1}{\tau}\left(\epsI_E\left(\TI\right) -\epsI_E\left(T\right)\right),\\
    &\frac{{d\sigma}_{\langle ij\rangle}}{dt} = -\frac{1}{\tau_\sigma}\sigma_{\langle ij\rangle}.   
\end{split}
\end{align}

\smallskip 

\begin{theorem}\label{theo:2}
The tensor $\TT_{ij}$ compatible with the requirement \eqref{prod11} has the following form:
\begin{align}\label{TT1}
	\TT_{ij} = \frac{2}{3}\theta \epsK_E\left(T\right) \, \delta_{ij} + \left(1 - \theta\right)\left\{\nu\frac{\pp_{ij}}{\rho} + \frac{2}{3}\left(1-\nu\right)\epsK_E\left(\TK\right)\, \delta_{ij}\right\},
\end{align}
where  $\pp_{ij} = -t_{ij}$ is the pressure tensor and the two parameters $\theta$ and $\nu$ are related to $\tau$, $\tau_\sigma$ and $\tes$ by
\begin{align} \label{eq:tau-param}
	\frac{1}{\tau} = \frac{\theta}{\tes}, \qquad
    \frac{1}{\tau_\sigma} = \frac{1}{\tes}\left[1-\nu \left(1 - \theta\right)\right].
\end{align}
Since $\TT_{ij}$ is definite positive, the ranges of these parameters are restricted to $\nu \in \left[ -\frac{1}{2},\, 1\right]$ and $\theta \in \left[0,\, 1\right]$. 
\end{theorem}
\begin{proof}
By substituting Eq.~\eqref{eq:stress} with Eq.~\eqref{stress} and Eq.~\eqref{momG11} into the left-hand-side of Eq. \eqref{prod11}, we have
\begin{align}
	 \left(
	\begin{matrix}
		\displaystyle
    \rho\TT_{ll} - \pp_{ll} \\ %[7pt] \displaystyle
    \rho \TT_{\langle ij \rangle} - \pp_{\langle ij \rangle} %\\[10pt]
	\end{matrix}
	\right)
	= 
    \left(
	\begin{matrix}
		\displaystyle 
        -2 \frac{\tes}{\tau}\rho \left(\epsK_E\left(\TK\right)-\epsK_E\left(T\right)\right) \\ 
        \displaystyle 
        \frac{\tes}{\tau_\sigma} \sigma_{\langle ij\rangle} 
	\end{matrix}
	\right).
 \label{prod11r}
\end{align}
Recalling that $\pp_{ll} = 3 \PP = 2\rho \epsK_E\left(\TK\right)$ and $\pp_{\langle ij \rangle} = - \sigma_{\langle ij \rangle}$, we have
\begin{align}
    &\TT_{ll} = 2 \left(1- \frac{\tes}{\tau}\right)\epsK_E\left(\TK\right)+2\frac{\tes}{\tau}\epsK_E\left(T\right),\\
    &\TT_{\langle ij\rangle} = \frac{1}{\rho}\left(-1 + \frac{\tes}{\tau_\sigma}\right) \sigma_{\langle ij\rangle}.
\end{align}
In order to have correspondence with previous studies \cite{Andries-2000, Brull-2009, Kosuge-2019}, the parameters $\theta$ and $\nu$ are defined as
\begin{align}
\begin{split}
	\frac{\tes}{\tau} = \theta, \qquad
    \frac{\tes}{\tau_\sigma} = 1-\nu \left(1 - \theta\right),    
\end{split}
\label{eq:T2}
\end{align}
which provide Eq.~\eqref{eq:tau-param}.
With these parameters, from Eq.~\eqref{eq:T2}, $\TT_{ij} = \TT_{ll}\delta_{ij}/3 + \TT_{\langle ij\rangle}$ has the form of Eq.~\eqref{TT1}. 

Contrasting with the macroscopic-level determination of  $\TT_{ij}$, expressed in Eq.~\eqref{TT1}, the microscopic expression of $\TT_{ij}$ defined  in Eq.~\eqref{momG11} (with  $\G$ in place of $\bar{\G}$), ensures the definite positiveness of $\TT_{ij}$. This difference between the macroscopic and microscopic descriptions is recognized in the literature as the issue of {\it realizability} \cite{Hamburger-1944}. To maintain consistency in these different levels of description, the parameter ranges of $\nu$ and $\theta$ are restricted.
Let us rewrite $\TT_{ij}$ defined in Eq.~\eqref{eq:Tijp} as follows
\begin{align}
	&\TT_{ij} = \theta \frac{p}{\rho}\delta_{ij} +  \left(1 - \theta\right) \frac{\RR_{ij}}{\rho} ,
\end{align}
with
\begin{align} \label{eqR}
	\RR_{ij} =  \nu \pp_{ij} + \left(1-\nu\right)\PP \delta_{ij}.
\end{align}
We may notice that $\TT_{ij}$, $\RR_{ij}$, and $\pp_{ij}$ have diagonal form. Let $\lambda_i^\TT$, $\lambda_i^\RR$, $\lambda_i^\pp$ ($i=1,2,3$) be the eigenvalues of, respectively, $\TT_{ij}$, $\RR_{ij}$, and $\pp_{ij}$. From Eq.~\eqref{eqR}, we have
\begin{align} \label{eq:lambda_R}
	\lambda_i^\RR = \nu \lambda_i^\pp + \left(1-\nu\right)\PP . 
\end{align}
Since $\PP  = \pp_{ll}/3 = \left(\lambda_1^\pp + \lambda_2^\pp + \lambda_3^\pp\right)/3$, Eq.~\eqref{eq:lambda_R} can be rewritten as follows:
\begin{align}
	\lambda_i^\RR = \frac{1+2\nu}{3} \lambda_i^\pp + \frac{1-\nu}{3}\left(\lambda_j^\pp + \lambda_k^\pp\right), \qquad \left(i\neq j\neq k\right). 
\end{align}
Sufficient condition for $\RR_{ij}$ to be positive definite is $-\frac{1}{2} \leq \nu \leq 1$.
Similarly, we have
\begin{align}
	\lambda_i^\TT = \theta \frac{p}{\rho} + \left(1 - \theta\right)\frac{\lambda_i^\RR}{\rho},
\end{align}
from which it is seen that the sufficient condition for $\TT_{ij}$ to be positive definite is $0\leq \theta \leq 1$, in addition to $-\frac{1}{2} \leq \nu \leq 1$.
\end{proof}

\smallskip 

\begin{corollary}\label{coro}
Defining the relaxation energy of the translational mode $\eKrel$, in analogy to $\eIrel$ given in Eq.~\eqref{momG11I}, as:
\begin{align}
    \eKrel = \frac{1}{\rho} \inta \frac{m}{2}C^2\, \G \phiI \da, \label{momG11K}
\end{align}
the following relation holds
\begin{align}\label{def:TKrel}
	\eKrel = \left(1 - \theta\right) \epsK_E\left(\TK\right) + \theta \epsK_E\left(T\right),
\end{align}
and we  have
\begin{align}\label{TKIrel}
	\eps = \epsK_E\left(\TKrel\right) + \epsI_E\left(\TIrel\right),
\end{align}
where the relaxation temperature of translational mode $\TKrel$ is defined by
\begin{align}
    \TKrel = \epsKinv_E\left(\eKrel\right).
\end{align}
Similarly, $\eIrel$ satisfies
\begin{align}\label{def:TIrel}
	\eIrel = \left(1-\theta\right) \epsI_E\left(\TI\right) + \theta\, \epsI_E\left(T\right).
\end{align}
\end{corollary}
\begin{proof}
From Eq.~\eqref{momG11}, we have $\eKrel = \TT_{ll}/2$. Then, by taking the trace part of Eq.~\eqref{TT1}, we have Eq.~\eqref{def:TKrel}. Since Eq.~\eqref{momG5}$_3$ is sum of Eq.~\eqref{momG11K} and Eq.~\eqref{momG11I}, we have Eq.~\eqref{TKIrel}. After subtracting Eq.~\eqref{TKIrel} from Eq.~\eqref{def:TKrel}, and taking into account Eq.~\eqref{temperatures}, we obtain Eq.~\eqref{def:TIrel}.
\end{proof}

\begin{remark}
The tensor $\TT_{ij}$ can also be expressed with the functional form of the pressure $p\left(\rho,T\right) = 2\rho \epsK_E\left(T\right)/3$,  as follows:
\begin{align}
	\rho \TT_{ij} = \theta p\delta_{ij} + \left(1 - \theta\right)\left\{\nu \pp_{ij} + \left(1 - \nu\right) \PP \delta_{ij}\right\}.
 \label{eq:Tijp}
\end{align}
Moreover, from the Corollary~\ref{coro}, it is seen that the tensor $\TT_{ij}$ can be expressed as 
\begin{align}\label{def:Tij}
	\TT_{ij} = \frac{2}{3}\epsK_E\left(\TKrel\right) \, \delta_{ij} - \nu \left(1 - \theta\right)\frac{\sigma_{\langle ij\rangle}}{\rho}.
\end{align}
\end{remark}

\smallskip

\begin{remark}
    Eq. \eqref{def:TIrel} and Eq. \eqref{def:TKrel}, given that $\theta \in \left[0,\, 1\right]$, define $\eIrel$ and $\eKrel$ as convex combinations, respectively, of $\epsI_E(\TI)$ and $\epsI_E(T)$, and of $\epsK_E(\TK)$ and $\epsK_E(T)$. 
    On the other hand, from Eq. \eqref{eq:Tijp} it is seen that $\TT_{ij}$ is a convex combination of $p\left(\rho,T\right)\delta_{ij}$ and $\nu \pp_{ij} + \left(1 - \nu\right) p(\rho, \TK)\delta_{ij}$, but the latter is not a convex combination of $\pp_{ij}$ and $p(\rho, \TK)\delta_{ij}$ since $\nu \in \left[ -\frac{1}{2},\, 1\right]$.
\end{remark}

\smallskip 

\smallskip 

\begin{remark}
    The difference between the present model and the model proposed by Kosuge et. al. \cite{Kosuge-2019} 
   is not limited to the state density $\phiI$ and the normalization function $A\left(T\right)$, but also involves the definition of $\TIrel$ and the introduction of $\TKrel$. In the previous model, $\TIrel$ is introduced as a convex combination of $T$ and $\TI$, i.e., $\TIrel = \theta T + \left(1-\theta\right) \TI$ with $\theta \in \left[0,\,1\right]$. On the other hand, here, the relaxation temperatures are defined through the energy as shown in Eq.~\eqref{TKIrel}. The two definitions of $\TIrel$ coincide in the case of polytropic gases. These definitions of the relaxation temperatures $\TKrel$ and $\TIrel$ also appear in \cite{Dauvois-2021, Mathiaud-2022}.    
\end{remark}

\smallskip

\subsection{H-Theorem and properties of the novel ES-BGK model}

Over the parameter domain $\nu \in \left[ -\frac{1}{2},\, 1\right]$ and $\theta \in \left[0,\, 1\right]$, which ensures that $\TT_{ij}$ is positive-definite, we can prove the H-Theorem:

\smallskip

\begin{theorem} \label{prop:1}
    The Boltzmann equation \eqref{eq:Boltzmann}, with the collisional term given in Eqs.~\eqref{eq:Q-ESBGK}, \eqref{def:G}, \eqref{eq:GE}, and \eqref{TT1} satisfies the H-theorem:
\begin{equation}
    \Sigma = - \frac{1}{\tes}\kB \inta \left(\G- f\right) \log f\, \phiI \da \geq 0.
\end{equation}    
\end{theorem}

\begin{proof} 
From the Boltzmann equation, by taking the moment of Eq.~\eqref{hgen}, we have the entropy balance law
\begin{align}
	\frac{\partial h}{\partial t} + \frac{\partial h_i}{\partial x_i} = \Sigma,
\end{align}
where the entropy density $h$ is defined in Eq.~\eqref{hgen}, and the entropy flux $h_i$ and production $\Sigma$ are defined as follows:
\begin{align}
\begin{split}
	&h_i = \inta \xi_i\, H\left(f\right)\, \phiI \da,\\
	&\Sigma = - \frac{1}{\tes}\kB \inta \left(\G- f\right) \log f\, \phiI \da = \frac{1}{\tes} \inta \left(\G -f\right) H'\left(f\right)\, \phiI \da. \label{entropyprod}
\end{split} 
\end{align}
Since $H\left(f\right)$ is a concave function, we have
\begin{align}\label{fgd}
	\left(\G - f\right) H'\left(f\right) \geq H\left(\G\right) - H\left(f\right).
\end{align}
From Eq.~\eqref{entropyprod}, taking into account \eqref{fgd} the following inequality holds 
\begin{align} \label{eq:Sigma11}
	\Sigma \geq \frac{1}{\tes}\inta \left(H\left(\G\right) - H\left(f\right)\right)\, \phiI \da = \frac{1}{\tes}\left(h^\G-h\right).
\end{align}
Let $h^{(11)}$ be the maximized entropy under the constraints that the first eleven moments of $f$ are $\left(\rho,\, \rho v_i,\, \rho v^2 + 2 \rho \eps,\, \rho v_i v_j + \pp_{ij}\right)$ (see also Remark \ref{rem:11}). These eleven moments correspond to the substitution of $\G$ with $f$ in Eq.~\eqref{momG5} and Eq.~\eqref{momG11}, which results in replacing the macroscopic  quantities from $\TT_{ij}$ to $\pp_{ij}/\rho$ and from $\TIrel$ to $\TI$. Similarly to the derivation of $\G$ shown in Eq.~\eqref{def:G} and Eq.~\eqref{eq:GE}, we obtain a distribution function $f^{(11)}$ that maximizes the entropy density for eleven moments:
\begin{align}
	f^{(11)} = \frac{\rho }{m\left(2\pi\right)^{3/2} \left[\det \left(\pp/\rho\right)\right]^{1/2} A\left(\TI\right)} \exp \left\{-\frac{1}{2} \left(\xi_i - v_i\right) \left(\left(\pp/\rho\right)^{-1}\right)_{ij} (\xi_j - v_j) -\frac{I}{\kB \TI}\right\},
 \label{eq:f11}
\end{align}
 and then we obtain the maximized entropy $h^{(11)}$ from Eq.~\eqref{hgen} as follows:
\begin{align}\label{h11}
	h^{(11)} = - \kBm \rho \left(\log \frac{\rho}{m\left(2\pi\right)^{3/2}\sqrt{\det \left(\pp/\rho\right)}A\left(\TI\right)} - \frac{m\epsI_E\left(\TI\right)}{\kB \TI} - \frac{3}{2}\right).
\end{align}
For any number of truncation $N$, using the MEP, the entropy $h^{(N)}$ that is maximized under the constraints that the first $N$ moments are prescribed satisfies the inequality
\begin{equation}\label{hmaxx}
    h^{(N)}\geq  h
\end{equation}
(see Appendix \ref{App:Max} for the proof), therefore in particular we have:  
\begin{equation} \label{eq:h11>h}
    h^{(11)} \geq h.    
\end{equation}
Moreover, we can prove that $h^\G \geq h^{(11)}$. First, we note that $A\left(T\right)$ given in Eq.~\eqref{eq:AT} may be expressed with the specific heat of internal mode, $c_v^I(T)=\frac{\partial \epsI_E(T)}{\partial T}$, as follows:
\begin{align} \label{eq:AT11}
	A\left(T\right) = A_0 \exp\left\{-\frac{m}{\kB}\left(\frac{\epsI_E\left(T\right)}{T} - \frac{\epsI_E\left(T_*\right)}{T_*}\right) + \frac{m}{\kB} \int_{T_*}^{T} \frac{c_v^I\left(\tau\right)}{\tau}d\tau\right\}.
\end{align}
From Eq.~\eqref{eq:AT11} and the explicit expressions of $h^{(11)}$ and $h^\G$, given respectively in Eq.~\eqref{h11} and Eq.~\eqref{hGexpression}, we have
\begin{align}
	h^{(11)} - h^\G &= - \kBm \rho \left(\log \sqrt{\frac{\det \TT}{\det \left(\pp/\rho\right)}} + \log \frac{A\left(\TIrel\right)}{A\left(\TI\right)} - \frac{m\epsI_E\left(\TI\right)}{\kB \TI} + \frac{m\epsI_E\left(\TIrel\right)}{\kB\TIrel}\right)\\
	&= \frac{1}{2}\kBm \rho \log \frac{\det\left(\pp/\rho\right)}{\det \TT} + \rho \int_{\TIrel}^{\TI}\frac{c_v^I\left(T'\right)}{T'}dT'. \label{eq:h11hG}
\end{align}
Given the following inequality:
\begin{align}\label{det_inequality}
	\frac{\det\left(\pp/\rho\right)}{\det \TT} \leq \left(\frac{\epsK_E\left(\TK\right)}{\epsK_E\left(\TKrel\right)}\right)^3,
\end{align}
which is proven in Appendix \ref{App:det}, and introducing the specific heat of the translational mode, $c_v^K = \frac{\partial \epsK_E(T)}{\partial T} = \frac{3}{2}\frac{k_B}{m}$,   Eq.~\eqref{eq:h11hG} satisfies
\begin{equation}
    \begin{split}
	h^{(11)} - h^\G \leq & 
	\frac{1}{2}\kBm \rho \log \left(\frac{\epsK_E \left(\TK\right)}{\epsK_E \left(\TKrel\right)}\right)^3 + \rho \int_{\TIrel}^{\TI}\frac{c_v^I\left(T'\right)}{T'}dT'\\
	& = \rho \int_{\TKrel}^{\TK} \frac{c_v^K}{T'}dT' + \rho \int_{\TIrel}^{\TI} \frac{c_v^I\left(T'\right)}{T'}dT'\\
	& = \rho \int_{\eKrel}^{\epsK} \frac{1}{\epsKinv_E\left(\eps^{K'}\right)}d\eps^{K'}  + \rho \int_{\eIrel}^{\epsI} \frac{1}{\epsIinv_E \left(\eps^{I'}\right)}d\eps^{I'} \\
	& = \rho s\left(\rho, \epsK, \epsI\right) - \rho s\left(\rho, \eKrel, \eIrel\right),
    \end{split}
\end{equation}
where we have adopted Eq.~\eqref{eq:epsK}, and $s\left(\rho, \epsK, \epsI\right)$ is a function which satisfies the following generalized Gibbs relation \cite{Arima-2017}:
\begin{align}
	ds\left(\rho, \epsK, \epsI\right) = \frac{1}{\epsKinv_E\left(\epsK\right)}d\epsK + \frac{1}{\epsIinv_E\left(\epsI\right)}d\epsI - \frac{k_B}{m} \frac{1}{\rho}d\rho.
\end{align}
It remains to be proven that $s\left(\rho, \epsK, \epsI\right) \leq s\left(\rho, \eKrel, \eIrel\right)$, and we follow the procedure proposed in \cite{Dauvois-2021}. To this aim, recalling that $\eKrel$ and $\eIrel$ depend on $\theta$ (see Eq.~\eqref{def:TKrel} and Eq.~\eqref{def:TIrel}), we introduce 
\begin{align}
	S\left(\theta\right) = s\left(\rho, \eKrel, \eIrel\right),
\end{align}
which satisfies
\begin{align}
	S\left(0\right) = s\left(\rho, \epsK, \epsI\right).
\end{align}
The function $S\left(\theta\right)$ is a concave function of $\theta$ because we have
\begin{equation}
\begin{split}
	\frac{\partial S}{\partial \theta}\left(\theta\right) &= \frac{\partial s\left(\rho, \eKrel, \eIrel\right)}{\partial \eKrel}\frac{\partial \eKrel}{\partial \theta}
	+ \frac{\partial s\left(\rho, \eKrel, \eIrel\right)}{\partial \eIrel}\frac{\partial \eIrel}{\partial \theta} \\
 &=\frac{1}{\TKrel}\left(\eps_E^K\left(T\right) - \eps_E^K\left(\TK\right)\right) + \frac{1}{\TIrel} \left(\eps_E^I\left(T\right) - \eps_E^I\left(\TI\right)\right),\\
\end{split}
\end{equation}
\begin{equation}
\begin{split}
	\frac{\partial^2 S}{\partial \theta^2}\left(\theta\right) = - \frac{1}{c_v^K {\TKrel}^2}\left(\eps_E^K\left(T\right)-\eps_E^K\left(\TK\right)\right)^2 
	- \frac{1}{c_v^I\left(\TIrel\right) {\TIrel}^2}\left(\epsI\left(T\right)-\epsI\left(\TI\right)\right)^2 \leq 0,
 \end{split}
\end{equation}
where in the last inequality, we have used $c_v^K  \geq 0$ and $c_v^I(\TIrel) \geq 0$. Moreover, since  $\TKrel=T$ and $\TIrel=T$ when $\theta=1$, we have $\frac{\partial S}{\partial \theta}\left(1\right) = 0$ from Eq.~\eqref{temperatures}. Therefore, $S$ is an increasing function of $\theta$ on the interval $\left[0,\,1\right]$, and the following relation holds
\begin{align}
	S\left(0\right) \leq S\left(\theta\right).
\end{align}
Since $s\left(\rho, \epsK, \epsI\right) \leq s\left(\rho, \eKrel, \eIrel\right)$ is proven, we conclude that
\begin{align} \label{eq:hg>h11}
	h^\G \geq h^{(11)}.
\end{align}
Combining Eq.~\eqref{eq:hg>h11} with Eq.~\eqref{eq:h11>h}, we conclude that
\begin{align}
	h^\G \geq h,
\end{align}
and therefore, from Eq.~\eqref{eq:Sigma11}, it is proven that $\Sigma \geq 0$.
\end{proof}

Dividing the range of $\theta$ into $\theta \in \left(0,\, 1\right]$ and $\theta=0$, we have the following propositions.

\smallskip

\begin{proposition}\label{prop:2}
    When $\theta \in \left(0,\, 1\right]$, the distribution functions $f$ and $\G$ reduce to $f^{(E)}$ at the equilibrium (see Eq.~\eqref{Poly:q(T)}) where $Q\left(f\right) = 0$.
\end{proposition}

\smallskip

\begin{proposition}\label{prop:3}
    When $\theta = 0$, the distribution functions $f$ and $\G$ reduce to $f^{(6)}$, which is defined by
    \begin{equation} \label{eq:f6}
    	f^{(6)} =  \frac{\rho}{m \, A\left(\TI\right) } \left( \frac{m}{2 \pi \kB \TK} \right)^{3/2} 
    	\exp \left( -\frac{m C^2}{2\kB \TK} - \frac{I}{\kB \TI} \right) ,
    \end{equation}
    at the equilibrium where $Q\left(f\right)=0$.
\end{proposition}

\begin{proof}[Proof of Proposition~{\upshape\ref{prop:2}}]
We follow the procedure outlined in \cite{Kosuge-2019}. At an equilibrium state where $Q=0$, we have $f=\G$.  
Then, from Eq. \eqref{momG11K} and Eq. \eqref{momG11}$_2$ we have
\begin{align}
    \TKrel = \TK, \qquad 
    \TIrel = \TI.
\end{align}
Moreover, since $\TT_{ij} = \pp_{ij}/\rho$ from Eq. \eqref{eq:stress} and Eq. \eqref{momG11},  Eq. \eqref{eq:Tijp} provides
\begin{align}
    \left(1 - \nu + \theta \nu\right) \pp_{ij} = \left\{ \theta p + \left(1 - \theta\right) \left(1 - \nu\right) \PP \right\} \delta_{ij}.
\end{align}
This indicates that $\sigma_{\langle ij\rangle} = -\rho \TT_{\langle ij\rangle} = 0$ and that $\theta \left(\TK - T \right) = 0$. The latter relation, being $\theta\neq 0$, gives $\TK = T$. Similarly, from Eq.~\eqref{def:TIrel}, we obtain $\TI = T$. Therefore, the following relation holds
\begin{align} \label{Teq}
	\TK = \TI = \TKrel = \TIrel = T,
\end{align}
and $\TT_{ij}=\kBm T\delta_{ij}$. From Eq. \eqref{def:G}, $\G=f^{(E)}$ and then $f=f^{(E)}$. 

Inversely, assuming $f = f^{(E)}$, Eq. \eqref{def:noneqT} with Eq. \eqref{Poly:EnergyDefEK} and Eq. \eqref{Poly:EnergyDefEI} gives $\TK = \TI =T$. 
Then, Eq.~\eqref{def:TKrel} and Eq.~\eqref{def:TIrel} provide Eq.~\eqref{Teq}. Recalling $\sigma_{\langle ij\rangle} = 0$ in this case, from Eq.~\eqref{eq:Tijp} we have $\TT_{ij} = p / \rho \delta_{ij}$. Therefore, $\G = f^{(E)}$, and then $f = \G$, which provides $Q=0$.
\end{proof}

\begin{proof}[Proof of Proposition~{\upshape\ref{prop:3}}]
Since $\theta = 0$, we have $\TK = \TKrel$ and $\TI = \TIrel$ from Eq.~\eqref{def:TKrel} and Eq.~\eqref{def:TIrel}, respectively. Therefore, we have
\begin{align}
    \inta m \xi^2 \left(\G-f\right) \phiI \da = 0, \qquad 
    \inta I \left(\G-f\right) \phiI \da = 0, 
\end{align}
which indicate that the collisional invariants are now $\left(m,\, m\xi_i,\, m\xi^2, I\right)$ (or $\left(m,\, m\xi_i,\, m(\xi^2+2{I}/{m}),\, m \xi^2\right)$ or  $\left(m,\, m\xi_i,\, m(\xi^2+2{I}/{m}),\, I\right)$). For the 6--moments $\left(\rho,\, \rho v_i,\, \rho v^2 + 2\rho \epsK,\, 2 \rho \epsI\right)$ (see Remark \ref{rem:6}) that correspond to the moments of the present collisional invariants, by exploiting MEP, we have Eq.~\eqref{eq:f6}.

In an equilibrium state, for which $Q = 0$, we have $f = \G$, which provides $\TT_{ij} = \pp_{ij}/\rho$ from Eq.~\eqref{eq:stress} and Eq.~\eqref{momG11}. Then,  Eq.~\eqref{eq:Tijp} gives $\pp_{ij} = \PP\delta_{ij}$ which results in $\sigma_{\langle ij\rangle} = 0$. Therefore, we have $\TT_{ij} = \kBm \TK \delta_{ij}$.
From Eq.~\eqref{def:G}, $\G = f^{(E)}$, and then $f=f^{(E)}$. 

Inversely, when we suppose $f=f^{(6)}$, we notice $\sigma_{\langle ij\rangle} = 0$. This results in, from Eq. \eqref{eq:Tijp}, $\TT_{ij} = \kBm \TK \delta_{ij}$. Therefore, $\G = f^{(6)}$, and then $f=\G$, which provides $Q=0$.
\end{proof}

\begin{remark}\label{rem:11}
The 11 moments of $f$, namely $\left(\rho,\, \rho v_i,\, \rho v^2 + 2 \rho \eps,\, \rho v_i v_j + \pp_{ij}\right)$, form the system of $11$ moments as specified by Eq.~\eqref{conserve} and Eqs.~\eqref{eksig}. By employing $f^{(11)}$ as presented in Eq.~\eqref{eq:f11}, we obtain  $\hat{H}^0_{ijk}=0$ and $q_i=0$.
\end{remark}

\smallskip

\begin{remark}\label{rem:6}
The system of the equations of 6 moments of $f$, that is, $\left(\rho,\, \rho v_i,\, \rho v^2 + 2 \rho \epsK,\, 2\rho \epsI\right)$, constitute Eq.~\eqref{conserve} and Eq.~\eqref{eksig}$_1$. With the use of $f^{(6)}$ given in Eq.~\eqref{eq:f6}, the constitutive functions are closed with $\hat{H}^0_{llk}=0$, $\sigma_{\langle ij\rangle} =0$ and $q_i=0$. See  \cite{Arima-2012-RET6,Arima-2016} for the closure of the present case.
\end{remark}

\subsection{Chapman-Enskog expansion} \label{sec:CE}

When the Knudsen number is small, one can formally derive the fluid-dynamic equations by means of the standard Chapman–Enskog procedure. Eq.~\eqref{eq:Boltzmann} with Eq.~\eqref{eq:Q-ESBGK} reduce, after straightforward calculations (see \cite{Andries-2000} for its details in the case of the ES-BGK model), to the Navier-Stokes-Fourier equations
\begin{align}
    \sigma_{\langle ij \rangle} = 2\mu \frac{\partial v_{\langle i}}{\partial x_{j\rangle}}, \qquad
    \Pi = - \mu_b \frac{\partial v_l}{\partial x_l}, \qquad
    q_i = -\kappa \frac{\partial T}{\partial x_i},
\end{align}
with the shear viscosity $\mu$, bulk viscosity $\mu_b$ and heat conductivity $\kappa$, given by
\begin{align}
    \mu = \frac{p\,\tes}{1-\nu + \theta \nu}, \qquad
    \mu_b = \frac{1}{\theta}\left(\frac{2}{3} - \frac{1}{\hat{c}_v}\right) p \, \tes, \qquad    
    \kappa = \kBm \left(1+\hat{c}_v\right) p \,\tes ,
    \label{eq:CE}
\end{align}
where $\hat{c}_v=m c_v/\kB$ is the dimensionless specific heat. 
These expressions of the transport coefficients are consistent with the ones obtained in \cite{Kosuge-2019}.
With these expressions, we can express the Prandtl number $Pr = c_p \mu/\kappa$, being $c_p = c_v + \kB/m$ the specific heat at constant pressure, as the function of the two parameters:
\begin{align}
    Pr = \frac{1}{1-\nu + \theta \nu}.
\end{align}
Moreover, the ratio of the viscosities is given by
\begin{align}
    \frac{\mu_b}{\mu} = \frac{1-\nu + \theta \nu}{\theta} \left(\frac{2}{3} - \frac{1}{\hat{c}_v}\right) = \frac{Pr}{\theta}\left(\frac{2}{3} - \frac{1}{\hat{c}_v}\right).
\end{align}
In this way, the transport coefficients are determined by $\theta$, $\nu$ and $\tes$ under a given value of $c_v$.
On the other hand, when the data of $c_v$, $\kappa$, $\mu$ and $\mu_b$ are available, we can evaluate the values of $\tes$, $\theta$, $\nu$. 
However, since data of $\mu_b$ are generally not available, we may set $\mu_b/\mu$ as an adjustable parameter \cite{Arima-2013-LW,Taniguchi-2014} (see also Sect.~\ref{sec:numerical}).

\section{Reduced ES-BGK model}\label{sec:reduced}

In order to reduce the computational cost of the numerical implementation of the ES-BGK model, the so-called \textit{reduced model} is usually introduced \cite{Dauvois-2021, Mathiaud-2022}. 
After defining the  marginal distribution functions $\Phi_m$ and $\Phi_I$  as follows:
\begin{align}
    \Phi_m\left(t, \xx, \cc\right) = \int_0^\infty m\,f\, \phiI dI, \qquad  
    \Phi_I\left(t, \xx, \cc\right) = \int_0^\infty I\, f\, \phiI dI,
    \label{eq:marginal}
\end{align}
and  introducing
\begin{align}
    \Psi_m\left(t, \xx, \cc\right) = \int_0^\infty m\,\G\, \phiI dI = m \G^{(K)}, \qquad 
    \Psi_I\left(t, \xx, \cc\right) = \int_0^\infty I\,\G\, \phiI dI = m\G^{(K)} \eIrel,
    \label{eq:marginalG}
\end{align}
the evolution equation of the marginal distribution functions $\mathbf{\Phi} \equiv \left(\Phi_m, \Phi_I\right)$ are obtained from the Boltzmann equation with $\mathbf{\Psi} \equiv \left(\Psi_m, \Psi_I\right)$ as follows:
\begin{align}
	\frac{\partial \boldsymbol{\Phi}}{\partial t} + \xi_i  \frac{\partial \boldsymbol{\Phi}}{\partial x_i}
	= \frac{1}{\tes} \left(\boldsymbol{\Psi} - \boldsymbol{\Phi}\right).
 \label{eq:Boltzmann-marginal}
\end{align}
The macroscopic fields are expressed as moments of $\Phi_m$ or $\Phi_I$ with respect to $\cc$ as follows:
\begin{gather}
	\rho = \int_{\R^3} \Phi_m\, d\cc, \qquad 
    \rho v_i = \int_{\R^3} \Phi_m \xi_i d\cc,\\
	\rho \epsK\left(\TK\right)  = \int_{\R^3}\frac{1}{2}\left(\xi_i-v_i\right)\left(\xi_i-v_i\right) \Phi_m d\cc, \qquad
	\rho \epsI\left(\TI\right)  = \int_{\R^3} \Phi_I d\cc, \\
	t_{ij} = - \int_{\R^3} (\xi_i-v_i)(\xi_j-v_j) \Phi_m d\cc, \\
	q_j = \int_{\R^3} \left\{\frac{1}{2}\left(\xi_i-v_i\right)\left(\xi_i-v_i\right) \Phi_m + \Phi_I \right\}\left(\xi_j-v_j\right) d\cc.
\end{gather}

\section{Study of standing planar shock waves} \label{sec:numerical}

A shock wave structure in one-space dimension is a traveling wave depending on $x_1$ and $t$ through $z=x_1-s \, t$, where $s$ is the shock velocity. As the Boltzmann equation is Galilean invariant, as usual we can consider the reference frame moving with the shock front for which $s=0$. Then, in order to investigate the structure of standing planar shock waves obtained with the novel ES-BGK model, Eq.~\eqref{eq:Boltzmann} is written in its steady one-dimensional form as follows:
\begin{equation} \label{eq:Boltzmann-1Dsteady}
	\xi_1 \frac{\partial f}{\partial x_1} = Q\left(f\right),
\end{equation}
and then suitably put in dimensionless form.

For a rarefied CO$_2$ gas, since the shear viscosity $\mu$ is well approximated by a power of the temperature \cite{Gilbarg-1953}, recalling Eq.~\eqref{eq:CE}$_1$ and following the notation in \cite{Kosuge-2019}, it is useful to write the relaxation time $\tes$ as
\begin{align}
    \tes = \frac{1}{\rho A_c\left(T\right)},
\end{align}
where the explicit expression of $A_c\left(T\right)$ as a power of $T$ will be given later (see Sect.~\ref{sec:numerical_results}).

\subsection{Dimensionless system}

Adopting the following dimensionless variables, as suggested in \cite{Kosuge-2018},
\begin{equation}
	\begin{gathered}
		\hatx = \frac{x_1}{L}, \qquad 
		\hatcx = \frac{\xi_1}{a_0}, \qquad 
		\hatvx = \frac{v_1}{a_0}, \qquad
		\hatrho = \frac{\rho}{\rho_0}, \\ 
		\hat{p} = \frac{p}{\rho_0 a_0^2/2}, \qquad
		\hatpij = \frac{\pp_{ij}}{\rho_0 a_0^2/2}, \qquad
		\hatTT_{ij} = \frac{\TT_{ij}}{\rho_0 a_0^2/2}, \qquad
        \hat\eps = \frac{\eps}{a_0^2/2}, \\
        \hatepsK=\frac{\epsK}{a_0^2/2}, \qquad
        \hatepsI=\frac{\epsI}{a_0^2/2}, \qquad
        \hateKrel=\frac{\eKrel}{a_0^2/2},\qquad
        \hateIrel=\frac{\eIrel}{a_0^2/2}, \qquad 
        \hatq_1 = \frac{q_1}{\rho_0 a_0^4/2}, \\
		\hatT = \frac{T}{T_0}, \qquad 
		\hatTK = \frac{\TK}{T_0}, \qquad 
		\hatTI = \frac{\TI}{T_0}, \qquad
        \hatTKrel = \frac{\TKrel}{T_0}, \qquad
		\hatTIrel = \frac{\TIrel}{T_0}, \\
		\hatf = \frac{m A\left(T_0\right)f}{\rho_0 a_0^{-3}}, \qquad
		\hatGE = \frac{m \GE}{\rho_0a_0^{-3}}, \qquad
		\hatGI = A\left(T_0\right)\GI, \qquad
		\hatI = \frac{I}{m a_0^2/2}, \\
        \hatAc\left(\hatT\right) = \frac{A_c\left(T\right)}{A_c\left(T_0\right)}, \qquad
        \hatA\left(\hatT\right) = \frac{A\left(T\right)}{A\left(T_0\right)}, \qquad 
        \hatphiI = \frac{m a_0^2 \phiI}{2 A\left(T_0\right)},
	\end{gathered}
\end{equation}
where $\rho_0$ and $T_0$ are reference values for, respectively, the density and temperature; $a_0 = \left(2 \kBm T_0\right)^{1/2}$, $L = 2 a_0 / \left( \pi^{1/2} \rho_0 A_c\left(T_0\right)\right)$ is the mean free path of the gas molecules in the equilibrium state with density $\rho_0$ and temperature $T_0$, Eq.~\eqref{eq:Boltzmann-1Dsteady} takes the form
\begin{equation} \label{eq:Boltzmann-1Dsteady-hat}
	\hatcx \frac{\partial \hatf}{\partial \hatx} = \frac{2}{\pi^{1/2}} \hatQ\left(\hatf\right), \qquad
    \hatQ\left(\hatf\right) = \hatAc\left(\hatT\right) \hatrho \left( \hatG - \hatf\right),
\end{equation}
where $\hatG = \hatGE \hatGI$ with
\begin{gather}
	\hatGE = \frac{\hatrho}{\pi^{3/2} \left(\det\hatTT\right)^{1/2} } \exp \left\{-\left(\hat{\xi}_i - \hatvi\right) \left(\hatTT^{-1}\right)_{ij} \left(\hat{\xi}_j - \hatvj\right) \right\}, \\
	\hatGI = \frac{1}{\hatA\left(\hatTIrel\right)} \exp\left(-\frac{\hatI}{\hatTIrel}\right),
\end{gather}
and
\begin{equation}
	\begin{gathered}
		\hatrho = \int_\Rthree \int_0^\infty \hatf\, \hatphiI \dhatI \dhatcc, \\ 
		\hatvi = \frac{1}{\hatrho} \int_\Rthree \int_0^\infty \hatci \hatf\, \hatphiI \dhatI \dhatcc, \\
		\hatTT_{ij} = \left(1-\theta\right)\, \left[\left(1-\nu\right) \hatTK \delta_{ij} + \nu \frac{\hatpij}{\hatrho}\right] + \theta \,\hatT \delta_{ij}.
	\end{gathered} 
\end{equation}
Moreover, we have
\begin{gather}
	\hateps = \hatepsK + \hatepsI, \\
	\hatepsK = \frac{1}{2 \hatrho} \iint_0^{\infty} \abs{\hatcc - \hatvv}^2 \hatf\, \hatphiI \dhatI \dhatcc, \qquad
	\hatepsI = \frac{1}{\hatrho} \iint_0^{\infty} \hatI \hatf\, \hatphiI \dhatI \dhatcc,
\end{gather}
and
\begin{equation}
	\hateps_E\left(\hatT\right) =\int_{\hatT_*}^{\hatT} \hat c_v\left(\tau\right) \,d\tau, \qquad 
	\hatT = \hateps_E^{-1}\left(\hateps\right), \qquad 
    \hat{p} = \hatrho \hatT,
\end{equation}
where %$\hat c_v = c_v/ \kBm$, and 
$\hatT_* = T_*/T_0$.
The dimensionless translational temperature $\hatTK$ is readily given by 
\begin{equation}
	\hatTK = \frac{2}{3} \hatepsK,
 \label{eq:TKhat}
\end{equation}
while the dimensionless internal temperature $\hatTI$ and the dimensionless temperature $\hatTrel$ are obtained as implicit solutions of 
\begin{equation}
	\hatepsI_E\left(\hatTI\right) = \hatepsI, \qquad
	\hatepsI_E\left(\hatTIrel\right) = \theta\, \hatepsI_E\left(\hatT\right) + \left(1-\theta\right) \hatepsI_E\left(\hatTI\right).
 \label{eq:TIhat}
\end{equation}

\subsection{Similarity solution}

Since in the following we shall be interested in studying the structure of plane shock waves traveling along the $x_1$ direction (i.e. $\hatvy = \hatvz = 0$), it is useful to introduce the similarity solution of the form
\begin{equation}
	\hatf = \hatf\left( \hatx, \hatcx, \hatcr, \hatI \right), \qquad \hatcr = \left( \hatcy^2 + \hatcz^2\right)^{1/2}.
\end{equation}
Under this assumption, the distribution function $\hatGE$ is written as
\begin{equation}
	\hatGE = \frac{\hatrho}{\pi^{3/2} \left(\hatTT_{11}\right)^{1/2} \hatTT_{22}} 
	\exp \left(-\frac{\left(\hatcx - \hatvx\right)^2}{\hatTT_{11}} - \frac{\hatcr^2}{\hatTT_{22}}\right),
 \label{eq:hatGK}
\end{equation}
where the involved macroscopic quantities are written as follows:
\begin{equation}
	\begin{gathered}
		\hatrho = 2\pi \int_{0}^{\infty} \int_{-\infty}^{+\infty} \int_{0}^{\infty} \hatcr \hatf \hatphiI \dhatI \dhatcx \dhatcr, \\
		\hatvx = \frac{2\pi}{\hatrho} \int_{0}^{\infty} \int_{-\infty}^{+\infty} \int_{0}^{\infty} \hatcx \hatcr \hatf \hatphiI \dhatI \dhatcx \dhatcr, \\
		\hatp_{11} = 4\pi \int_{0}^{\infty} \int_{-\infty}^{+\infty} \int_{0}^{\infty} \hatcr \left(\hatcx - \hatvx\right)^2 \hatf \hatphiI \dhatI \dhatcx \dhatcr, \\
		\hatp_{22} = \hatp_{33} = 2\pi \int_{0}^{\infty} \int_{-\infty}^{+\infty} \int_{0}^{\infty} \hatcr^3 \hatf \hatphiI \dhatI \dhatcx \dhatcr, \\
		\hatTT_{11} = \theta\, \hatT + \left(1 - \theta\right) \left( \left(1 - \nu\right) \hatTK + \nu \frac{\hatp_{11}}{\hatrho} \right), \\
		\hatTT_{22} = \hatTT_{33} = \theta\, \hatT + \left(1 - \theta\right) \left( \left(1 - \nu\right) \hatTK + \nu \frac{\hatp_{22}}{\hatrho} \right), \\
		%\hatTT_{12} = \hatTT_{13} = \hatTT_{23} = 0,
		\hatepsK = \frac{2\pi}{\hatrho} \int_{0}^{\infty} \int_{-\infty}^{+\infty} \int_{0}^{\infty} \hatcr \left(\left(\hatcx - \hatvx\right)^2 + \hatcr^2\right) \hatf \hatphiI \dhatI \dhatcx \dhatcr, \\
		\hatepsI = \frac{2\pi}{\hatrho} \int_{0}^{\infty} \int_{-\infty}^{+\infty} \int_{0}^{\infty} \hatcr
        \hatI \hatf \hatphiI \dhatI \dhatcx \dhatcr, \\
        \hatq_1 = 2\pi \int_{0}^{\infty} \int_{-\infty}^{+\infty} \int_{0}^{\infty} \hatcr \left(\hatcx - \hatvx\right) \left( \left(\hatcx - \hatvx\right)^2 + \hatcr^2 + \hatI\right) \hatf \hatphiI \dhatI \dhatcx \dhatcr, 
	\end{gathered}
\end{equation}
and 
\begin{equation}
    \hatp_{ij}=0, \qquad \hatTT_{ij}=0 \quad \text{for} \quad i\neq j.
\end{equation}

\subsection{Reduced ES-BGK model for similarity solution}

In the present case, it is possible to introduce the marginal distribution function from Eq.~\eqref{eq:marginal} as follows:
\begin{align}
    &\phi_1 = 2\pi \int_{0}^{\infty}  \hatcr \, \hat\Phi_m \,  \dhatcr
    = 2\pi \int_{0}^{\infty}  \int_{0}^{\infty} \hatcr \hatf \hatphiI \dhatI \dhatcr
    , \\
    &\phi_2 = 2\pi \int_{0}^{\infty}  \hatcr^3 \, \hat\Phi_m \,  \dhatcr
    = 2\pi \int_{0}^{\infty}  \int_{0}^{\infty} \hatcr^3 \hatf \hatphiI \dhatI \dhatcr
    , \\
    &\phi_3 = 2\pi \int_{0}^{\infty}  \hatcr \, \hat\Phi_I \,  \dhatcr= 
    2\pi \int_{0}^{\infty}  \int_{0}^{\infty} \hatcr \hatI \hatf \hatphiI \dhatI \dhatcr
    , 
\end{align}
where $\hat\Phi_m = \Phi_m /\left(\rho_0a_0^{-3}\right)$ and $\hat\Phi_I = 2a_0\Phi_I /\rho_0$. Moreover, similarly to Eq.~\eqref{eq:marginalG}, we introduce
\begin{align}
 \begin{split}
    &\psi_1 = 2\pi \int_{0}^{\infty}  \hatcr \, \hat\Psi_m \,  \dhatcr
    =2\pi \int_{0}^{\infty}   \hatcr \, \hatGE \,  \dhatcr, \\
    &\psi_2 = 2\pi \int_{0}^{\infty}  \hatcr^3 \, \hat\Psi_m \,  \dhatcr, 
    =2\pi \int_{0}^{\infty}   \hatcr^3 \, \hatGE \,  \dhatcr, \\
    &\psi_3 = 2\pi \int_{0}^{\infty}  \hatcr \, \hat\Psi_I \,  \dhatcr
    =2\pi \hateIrel \int_{0}^{\infty}   \hatcr \, \hatGE \,  \dhatcr,      
 \end{split}
 \label{eq:marginalGhat}
\end{align}
where $\hat\Psi_m = \Psi_m /\left(\rho_0a_0^{-3}\right)$ and $\hat\Psi_I = 2a_0\Psi_I /\rho_0$.
Recalling Eq.~\eqref{eq:hatGK} and the Gaussian integrals $\int_{0}^{\infty} z \exp\left(-\frac{z^2}{\mu}\right) \, dz = \frac{\mu}{2}$ and $\int_{0}^{\infty} z^3 \exp\left(-\frac{z^2}{\mu}\right) \, dz = \frac{\mu^2}{2}$, Eq.~\eqref{eq:marginalGhat} are written as follows:
\begin{equation}
	\begin{split}
		\psi_1 
		&= \frac{2\hatrho}{\left(\pi \hatTT_{11}\right)^{1/2} \hatTT_{22}} \exp \left(-\frac{\left(\hatcx - \hatvx\right)^2}{\hatTT_{11}} \right) \int_{0}^{\infty} \hatcr \, \exp \left( - \frac{\hatcr^2}{\hatTT_{22}}\right) \dhatcr  
		=  \frac{\hatrho}{\left(\pi \hatTT_{11}\right)^{1/2}} \exp \left(-\frac{\left(\hatcx - \hatvx\right)^2}{\hatTT_{11}} \right),\\
		\psi_2 
		&= \frac{2\hatrho}{\left(\pi \hatTT_{11}\right)^{1/2} \hatTT_{22}} \exp \left(-\frac{\left(\hatcx - \hatvx\right)^2}{\hatTT_{11}} \right) \int_{0}^{\infty} \hatcr^3 \, \exp \left( - \frac{\hatcr^2}{\hatTT_{22}}\right) \dhatcr 
		=  \frac{\hatrho \hatTT_{22}}{\left(\pi \hatTT_{11}\right)^{1/2}} \exp \left(-\frac{\left(\hatcx - \hatvx\right)^2}{\hatTT_{11}} \right),\\
		\psi_3 
		&= \frac{2\hatrho \hateIrel}{\left(\pi \hatTT_{11}\right)^{1/2} \hatTT_{22}} \exp \left(-\frac{\left(\hatcx - \hatvx\right)^2}{\hatTT_{11}} \right) \int_{0}^{\infty} \hatcr \, \exp \left( -\frac{\hatcr^2}{\hatTT_{22}}\right) \dhatcr 
		=  \frac{\hatrho \hateIrel}{\left(\pi \hatTT_{11}\right)^{1/2}} \exp \left(-\frac{\left(\hatcx - \hatvx\right)^2}{\hatTT_{11}} \right),
	\end{split}
\end{equation}
and the following system is obtained from Eq.~\eqref{eq:Boltzmann-1Dsteady-hat} (see also Eq.~\eqref{eq:Boltzmann-marginal}):
\begin{equation} \label{eq:ESBGK-reduced}
	\hatcx \frac{\partial \phi_k}{\partial \hatx} = \frac{2}{\pi^{1/2}} \hatAc\left(\hatT\right) \hatrho \left( \psi_k - \phi_k\right), \qquad k = 1, 2, 3.
\end{equation}
It is also noted that the macroscopic quantities $\hatrho$, $\hatvx$, $\hatp_{11}$, and $\hatp_{22}$ involved in Eq.~\eqref{eq:ESBGK-reduced} may be written in terms of the marginal functions $\phi_1$, $\phi_2$, and $\phi_3$ as follows:
\begin{equation}
	\begin{gathered}
		\hatrho = \int_{-\infty}^{+\infty} \phi_1 \dhatcx, \qquad 
		\hatvx = \frac{1}{\hatrho} \int_{-\infty}^{+\infty} \hatcx \phi_1 \dhatcx, \\
		\hatp_{11} = 2 \int_{-\infty}^{+\infty} \left(\hatcx - \hatvx\right)^2 \phi_1 \dhatcx, \qquad
		\hatp_{22} = \int_{-\infty}^{+\infty} \phi_2 \dhatcx,
	\end{gathered}
\end{equation}
while the translational internal energy, $\hatepsK$, the internal energy associated to internal modes, $\hatepsI$, and the heat flux $\hatq_1$ are given by
\begin{equation}
	\hatepsK = \frac{1}{\hatrho} \int_{-\infty}^{+\infty} \left( \left(\hatcx - \hatvx\right)^2 \phi_1 + \phi_2 \right) \dhatcx, \qquad 
	\hatepsI = \frac{1}{\hatrho} \int_{-\infty}^{+\infty} \phi_3 \dhatcx.
\end{equation}
and
\begin{equation}
	\hatq_1 = \int_{-\infty}^{\infty} \left(\hatcx - \hatvx\right) \left( \left( \hatcx - \hatvx\right)^2 \phi_1 + \phi_2 + \phi_3\right) \dhatcx.
\end{equation}
The translational and internal temperatures, $\hatTK$ and $\hatTI$, are directly obtained from $\hatepsK$ and $\hatepsI$ from Eq.~\eqref{eq:TKhat} and Eq.~\eqref{eq:TIhat}, respectively.

\subsection{Numerical results \label{sec:numerical_results}}

In order to obtain the structure of planar shock waves for various values of the Mach number $M_0$, the system of integro-differential equations given in Eq.~\eqref{eq:ESBGK-reduced} is numerically solved on a one-dimensional finite computational domain.

Provided the quantities $\rho_0$, $v_{1,0}$, and $T_0$ representing, respectively, the density, $x_1$-component of the velocity, and temperature in the unperturbed equilibrium state ($x_1 \to -\infty$), the corresponding density, $\rho_1$, $x_1$-component of the velocity, $v_{1,1}$, and temperature, $T_1$, in the perturbed equilibrium state ($x_1 \to +\infty$) are obtained as a one-parameter solution of the Rankine-Hugoniot equations, being the Mach number $M_0$ the parameter.

In terms of dimensionless variables, the equilibrium distribution function in the unperturbed state ($x_1 \to -\infty$), $\hatf_0$, and in the perturbed state ($x_1 \to +\infty$), $\hatf_1$, are, respectively,
\begin{equation}
	\hatf_0 = \frac{\hatrho_0}{\left(\pi \hatT_0\right)^{3/2}} \exp\left( -\frac{\left(\hatcx - \hat v_{1,0}\right)^2 + \hatcr^2}{\hatT_0} \right) \frac{1}{\hatA\left(\hatT_0\right)} \exp\left( -\frac{\hatI}{\hatT_0} \right),
\end{equation}
and
\begin{equation}
	\hatf_1 = \frac{\hatrho_1}{\left(\pi \hatT_1\right)^{3/2}} \exp\left( -\frac{\left(\hatcx - \hat v_{1,1}\right)^2 + \hatcr^2 %+ \hatcy^2 + \hatcz^2
 }{\hatT_1} \right) \frac{1}{\hatA\left(\hatT_1\right)} \exp\left( -\frac{\hatI}{\hatT_1} \right),
\end{equation}
where $\hatrho_0 = \hatT_0 = 1$ due to the choice of the quantities $\rho_0$ and $T_0$ as reference values, respectively, for the density and the temperature in the definition of the dimensionless variables.

From the above expression of $\hatf_0$ and $\hatf_1$, the corresponding marginal distribution functions $\phi_{1,0}$, $\phi_{2,0}$, and $\phi_{3,0}$ in the unperturbed equilibrium state, and $\phi_{1,1}$, $\phi_{2,1}$, and $\phi_{3,1}$ in the perturbed equilibrium state are obtained ($i=0,1$):
\begin{equation} \label{eq:phi-bc}
	\begin{split}
		\phi_{1,i} &= \frac{\hatrho_i}{\left(\pi \hatT_i\right)^{1/2}} \exp\left( -\frac{\left(\hatcx - \hat v_{1,i}\right)^2}{\hatT_i} \right), \\
		\phi_{2,i} &= \frac{\hatrho_i \hatT_i}{\left(\pi \hatT_i\right)^{1/2}} \exp\left( -\frac{\left(\hatcx - \hat v_{1,i}\right)^2}{\hatT_i} \right), \\
		\phi_{3,i} &= \frac{\hatrho_i \hatepsI_i}{\left(\pi \hatT_i\right)^{1/2}} \exp\left( -\frac{\left(\hatcx - \hat v_{1,i}\right)^2}{\hatT_i} \right).
	\end{split}
\end{equation}
The previous expressions in Eq.~\eqref{eq:phi-bc} are used as boundary conditions in the process of numerically solving the system of equations outlined in Eq.~\eqref{eq:ESBGK-reduced}.

\def\scale{0.73}

\begin{figure}
	\centering
	\includegraphics[scale=\scale]{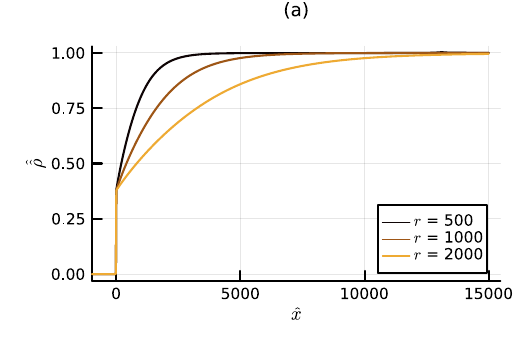}
	\includegraphics[scale=\scale]{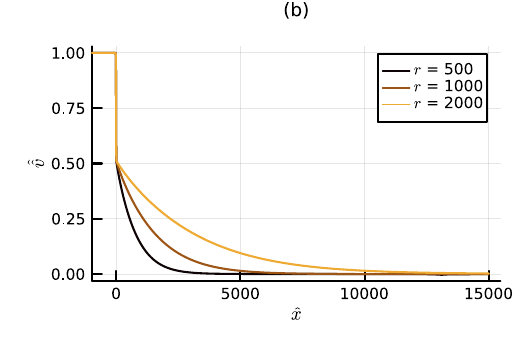} \\
    \includegraphics[scale=\scale]{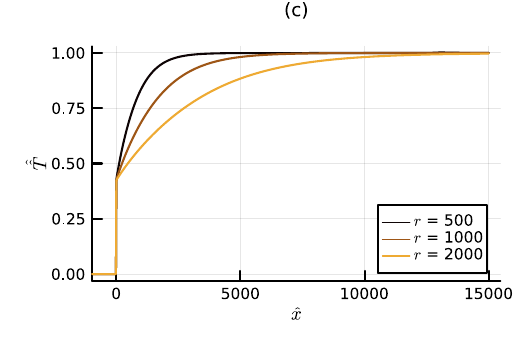}
	\includegraphics[scale=\scale]{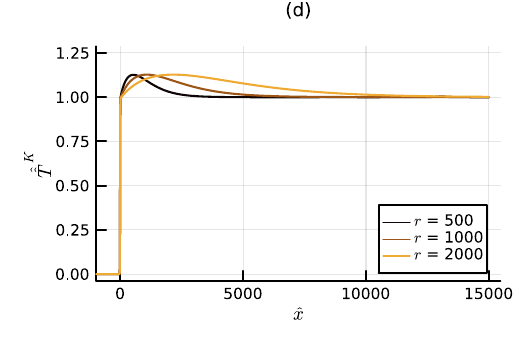} \\ 
    \includegraphics[scale=\scale]{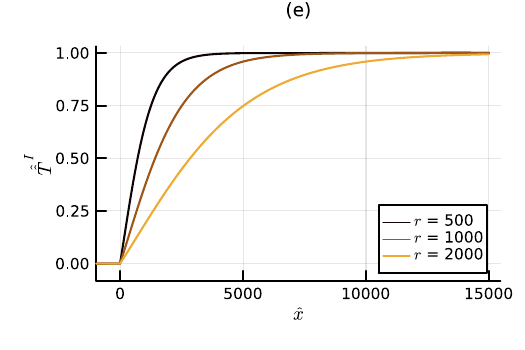}
	\includegraphics[scale=\scale]{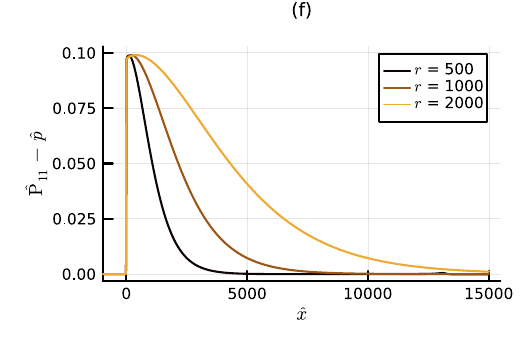} \\ 
    \includegraphics[scale=\scale]{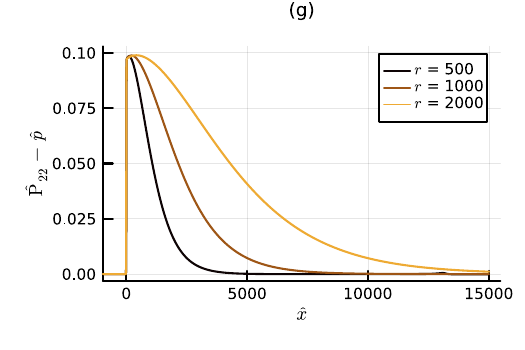}
	\includegraphics[scale=\scale]{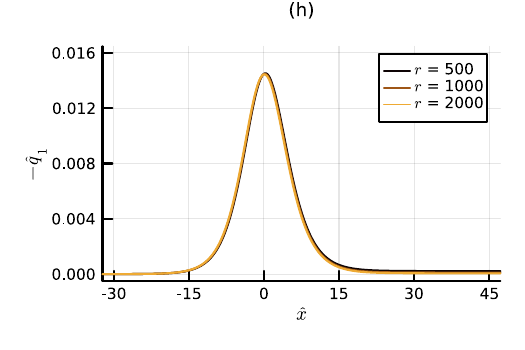} \\
	\caption{Profiles of the normalized density $\normrho$ (a); normalized velocity $\normv$ (b); normalized temperature $\normT$ (c); normalized translational temperature $\normTtr$ (d); normalized internal temperature $\normTint$ (e); dimensionless pressure difference $\hat{\pp}_{11} - \hat{p}$ (f); dimensionless pressure difference $\hat{\pp}_{22} - \hat{p}$ (g); dimensionless heat flux $-\hatq$ (h) for a plane shock wave corresponding to $M_0 = 1.3$, for three different values of the parameter $r = \left(\mu_b/\mu\right)_{T=T_0}$. \label{fig:M1.3} }
\end{figure}

\begin{figure}
	\centering
	\includegraphics[scale=\scale]{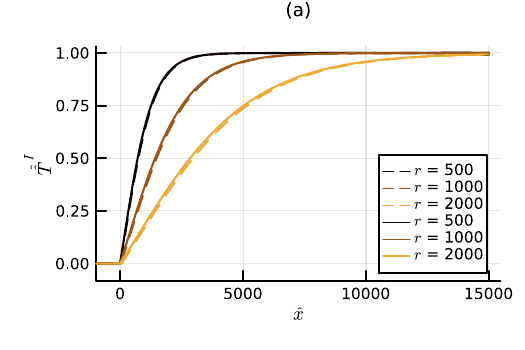}
	\includegraphics[scale=\scale]{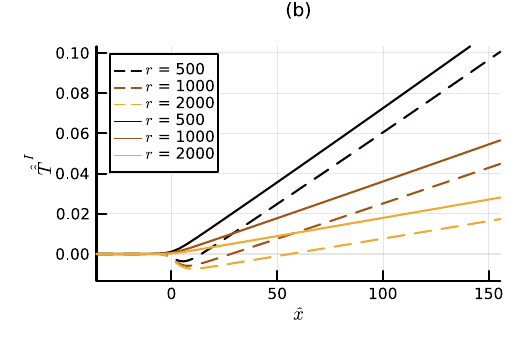} \\
	\caption{(a) Profiles of the normalized internal temperature, $\normTint$, for a plane shock wave corresponding to $M_0 = 1.3$, for three different values of the parameter $r = \left(\mu_b/\mu\right)_{T=T_0}$; (b) zoom of the profiles of the normalized internal temperature near the foot of the shock. Results obtained with the model presented in \cite{Kosuge-2019} are represented by dash lines; results obtained with the model presented in Sect.~\ref{sec:newmodel} are represented by solid lines. Panel (b) shows a zoom of the region \label{fig:M1.3-Tint} }
\end{figure}

\begin{figure}
	\centering
	\includegraphics[scale=\scale]{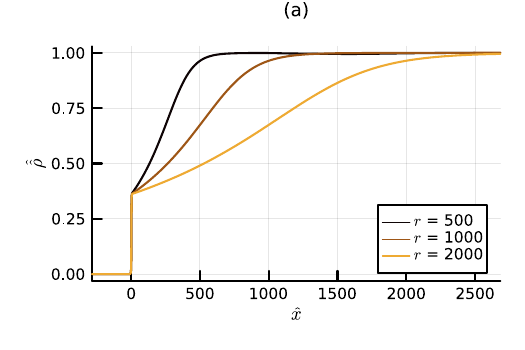}
	\includegraphics[scale=\scale]{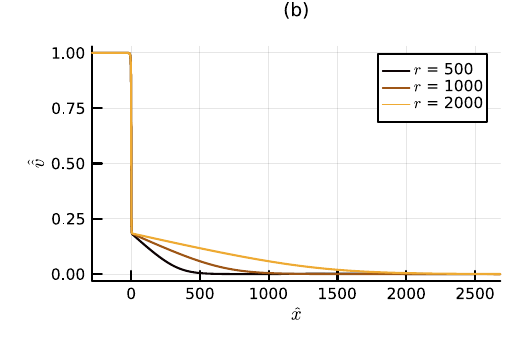} \\
    \includegraphics[scale=\scale]{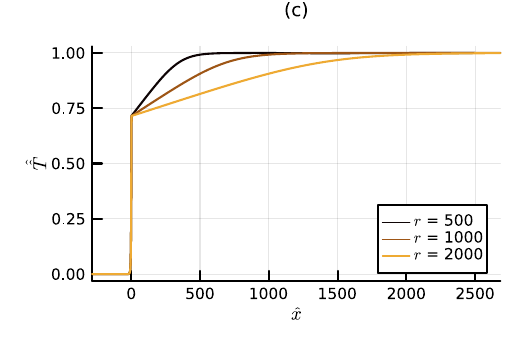}
	\includegraphics[scale=\scale]{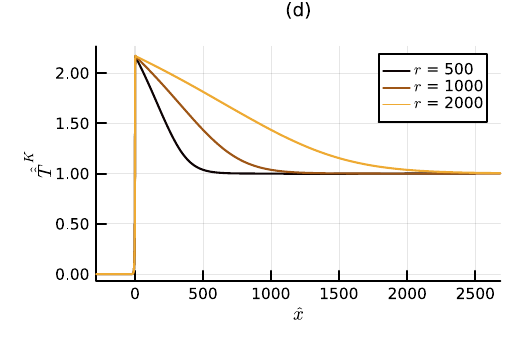} \\
    \includegraphics[scale=\scale]{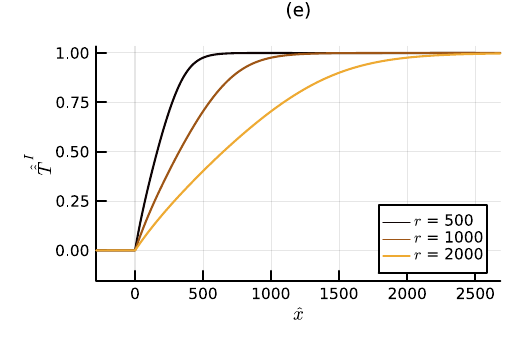}
	\includegraphics[scale=\scale]{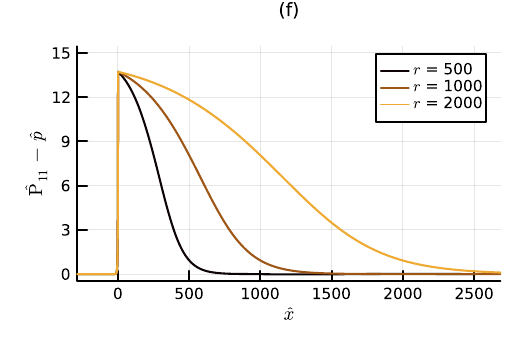} \\
    \includegraphics[scale=\scale]{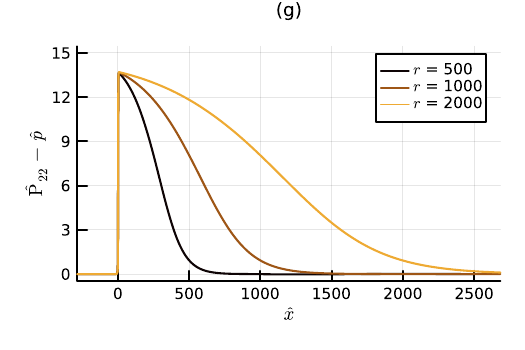}
	\includegraphics[scale=\scale]{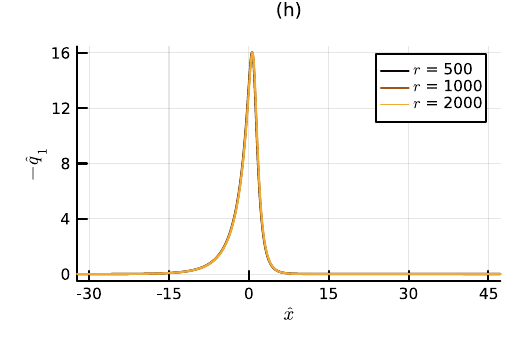} \\
	\caption{Profiles of the normalized density $\normrho$ (a); normalized velocity $\normv$ (b); normalized temperature $\normT$ (c); normalized translational temperature $\normTtr$ (d); normalized internal temperature $\normTint$ (e); dimensionless pressure difference $\hat{\pp}_{11} - \hat{p}$ (f); dimensionless pressure difference $\hat{\pp}_{22} - \hat{p}$ (g); dimensionless heat flux $-\hatq$ (h) for a plane shock wave corresponding to $M_0 = 5$, for three different values of the parameter $r = \left(\mu_b/\mu\right)_{T=T_0}$. \label{fig:M5} }
\end{figure}

\begin{figure}
	\centering
	\includegraphics[scale=\scale]{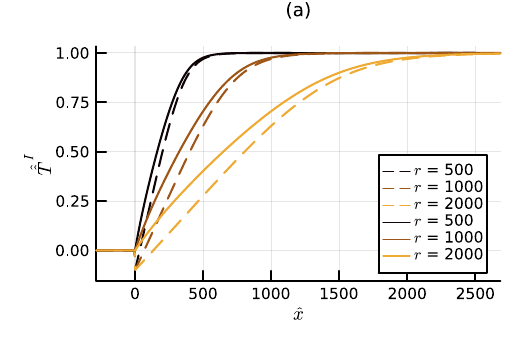}
	\includegraphics[scale=\scale]{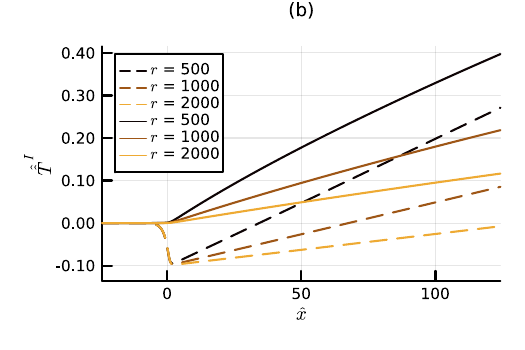} \\
	\caption{(a) Profiles of the normalized internal temperature, $\normTint$, for a plane shock wave corresponding to $M_0 = 5$, for three different values of the parameter $r = \left(\mu_b/\mu\right)_{T=T_0}$; (b) zoom of the profiles of the normalized internal temperature near the foot of the shock. Results obtained with the model presented in \cite{Kosuge-2019} are represented by dash lines; results obtained with the model presented in Sect.~\ref{sec:newmodel} are represented by solid lines.\label{fig:M5-Tint} }
\end{figure}

\smallskip 

In order to compare the results obtained by means of the model proposed by Kosuge et al. \cite{Kosuge-2019} to the model proposed here, calculations have been carried out adopting the same model parameters as those used in \cite{Kosuge-2019}, which in turn used model parameters discussed in \cite{Taniguchi-2014, Taniguchi-2016}. A carbon dioxide (CO$_2$) gas is considered, for which the temperature dependence of the specific heat, $\hat c_v = c_v / \left(\kB / m\right)$, may be approximated at around room temperature as follows \cite{Taniguchi-2014}:
\begin{equation} \label{eq:cv-poly}
	\hat c_v \left(T\right) = 1.412 + 8.697 \times 10^{-3} T - 6.575 \times 10^{-6} T^2 + 1.987 \times 10^{-9} T^3.
\end{equation}
The temperature dependence of the shear viscosity is approximated as $\mu \propto T^{0.935}$ \cite{Taniguchi-2014}. Therefore, from Eq.~\eqref{eq:CE}$_1$, we set $A_c\left(T\right)\propto T^{0.065}$, i.e., $\hat{A}_c\left(\hat{T}\right) = \hatT^{0.065}$. Following \cite{Kosuge-2019}, the values of $\nu$ and $\theta$ are suitably chosen as to match a value of the Prandtl number equal to 0.73 and a ratio $r$ of the bulk viscosity, $\mu_b$, and the viscosity, $\mu$, in the unperturbed equilibrium state varying in the range from 500 to 2000, as shown in Table~\ref{tab1}.
\begin{center}
\begin{table}
    \begin{tabular}{|c|c|c|} 
        \hline
        $r = \left( \mu_b/\mu \right)_{T=T_0} $ & $\nu$ & $\theta$ \\
        \hline
         $500$ & $-0.3702$ & $1.034 \times 10^{-3}$ \\ 
        $1000$ & $-0.3701$ & $5.169 \times 10^{-4}$ \\
        $2000$ & $-0.37$   & $2.585 \times 10^{-4}$ \\
        \hline
    \end{tabular}
    \caption{Values of $\nu$ and $\theta$ for $Pr = 0.73$ and $r = \left( \mu_b/\mu \right)_{T=T_0} = 500, 1000, 2000$. \label{tab1} }
\end{table}
\end{center}

\smallskip 

In order to facilitate the comparison of the results obtained by means of the two models, the profiles of the density, velocity, temperature, translational temperature and internal temperature in a planar shock wave, shown in Fig.~\ref{fig:M1.3} and Fig.~\ref{fig:M1.3-Tint} for $M_0 = 1.3$, and Fig.~\ref{fig:M5} and Fig.~\ref{fig:M5-Tint}  for $M_0 = 5$, are normalized, following \cite{Kosuge-2019}, as follows:
\begin{equation}
    \normrho = \frac{\hat{\rho} - \hat{\rho}_0}{\hat{\rho}_1 - \hat{\rho}_0}, \quad
    \normv = \frac{\hat{v}_1 - \hat{v}_{1,1}}{\hat{v}_{1,0} - \hat{v}_{1,1}}, \quad
    \normT = \frac{\hat{T} - \hat{T}_0}{\hat{T}_1 - \hat{T}_0}, \quad
    \normTtr = \frac{\hat{T}^K - \hat{T}_0}{\hat{T}_1 - \hat{T}_0}, \quad
    \normTint = \frac{\hat{T}^I - \hat{T}_0}{\hat{T}_1 - \hat{T}_0}.
\end{equation}

Despite the relevant differences in the model proposed in \cite{Kosuge-2019} and the novel model presented here, the numerical results obtained by means of the two models are very similar except for a remarkable difference in the profile of the internal temperature, $\TI$. Being the profiles of all the other macroscopic quantities very similar to those already published in \cite{Kosuge-2019}, the comparison is not reported here, and only the profiles obtained with the novel model presented in Sect.~\ref{sec:newmodel} are shown in Fig.~\ref{fig:M1.3} and Fig.~\ref{fig:M5}; the only comparison between the results obtained with the two models that we show pertains to the profile of the internal temperature, shown in Fig.~\ref{fig:M1.3-Tint} and Fig.~\ref{fig:M5-Tint}.

In Fig.~\ref{fig:M1.3-Tint}, corresponding to the case with $M_0 = 1.3$, it may be appreciated that the model proposed in \cite{Kosuge-2019} leads to a profile of the internal temperature dropping to values below the unperturbed one in the region close to the \textit{foot} of the shock profile.
In Fig.~\ref{fig:M5-Tint}, the same profile of the normalized internal temperature, $\normTint$, is shown for the case corresponding to a larger Mach number, $M_0 = 5$. In this case, values of the internal temperature lower than the unperturbed values (i.e. negative values of the normalized internal temperature $\normTint$) obtained by the model presented in \cite{Kosuge-2019} are even more noticeable than in the previous case shown in Fig.~\ref{fig:M1.3-Tint}.
In both cases, the profiles of the internal temperature obtained by means of the model proposed here are physically meaningful, since the profiles show that the internal temperature is monotonically non-decreasing through the shock profile.

As an additional consideration, it might be observed that the results presented in \cite{Kosuge-2019} pertaining the case $\theta = 0$ (i.e. $r \to \infty$), seem to show that the internal temperature, $\TI$, takes on values different from the unperturbed value of the temperature, $T_0$, across the shock structure (specifically, the results show that $\normTint < 0$, i.e. $\TI < T_0$ across the shock structure). On the other hand, the results presented here obtained with the newly developed model suggest that, as $r$ increases, the internal temperature $\TI$ across the shock structure tends to remain constant and equal to the unperturbed temperature $T_0$ (i.e. $\normTint = 0$ across the shock structure). The latter result is in agreement with the fact that $r \to \infty$ corresponds to the physical situation in which the internal molecular mode is \emph{frozen} and, as such, no variation in the internal temperature $\TI$ should be expected in the non-equilibrium region.

\section{Conclusions \label{sec:conclusions}}

In this study, we introduced a novel ES-BGK model of non-polytropic polyatomic gases that incorporates an internal state density function depending solely on the microscopic energy of internal modes and is, therefore, independent from the temperature. This model adheres to conservation laws and is capable of inducing the correct Prandtl number; moreover it upholds the H-theorem, distinguishing it from a model recently proposed in \cite{Kosuge-2019}. Additionally this model allows to obtain a closed system of macroscopic equations making use of the maximum entropy principle (MEP) in the spirit of Rational Extended Thermodynamics (RET).

We also introduced the so-called \textit{reduced version} of this model by incorporating marginal distribution functions. The numerical implementation of the reduced model enabled us to investigate the structure of planar shock waves in carbon dioxide (CO$_2$) and to make comparative assessments against results obtained from the previous model \cite{Kosuge-2019}. It is noteworthy that, for the reduced model and shock waves, we did not need to calculate $\phiI$ explicitly through the inverse Laplace transform. Nevertheless, for general solutions of the kinetic model, we must compute the expressions of $\phiI$, which can be challenging also numerically.

\smallskip 

Future studies will delve into areas not covered in this paper. In particular, they will include:
\begin{enumerate}[(i)]
\item The closure via MEP for this model is now possible and an evaluation of the production terms appearing in the macroscopic field equations obtained in the framework of RET using the collisional term proposed here;
\item An extension of the ES-BGK model proposed here in order to model separately the molecular internal modes of rotation and vibration. Preliminary investigations on this point can be found in the BGK model for collisional processes presented in \cite{Arima-2017}, and in the development of an ES-BGK model accommodating for rotational and discrete vibrational modes \cite{Dauvois-2021, Mathiaud-2022};
\item An analysis of the structure of standing planar shock waves in gases with a different interpolating function for the specific heat $c_v(T)$ than the one defined in Eq.~\eqref{eq:cv-poly}, and considering other interesting physical cases of $c_v(T)$ for different gases.
\end{enumerate}

\begin{appendices}
 
\section{MEP and proof of inequality (\ref{hmaxx}) \label{App:Max}}

We first recall a brief history of the maximum entropy principle (MEP) that was developed by Jaynes in the context of the theory of information \cite{Jaynes-1957a, Jaynes-1957b}. 

The applicability of MEP to nonequilibrium thermodynamics was originally proposed in 1967 by Kogan \cite{Kogan-1969}. A precise equivalence between MEP and RET, in the $13$ moment case, was proved in 1987 by Dreyer \cite{Dreyer-1987}; then, the MEP procedure was applied in 1993 by M\"uller and Ruggeri \cite{Muller-1993}, also for degenerate gases, to the general case of any number of moments, where it was proved for the first time that the closed system is symmetric hyperbolic if one chooses the Lagrange multipliers as field variables. The MEP was proposed again and popularized three years later by Levermore \cite{Levermore-1996}. The complete equivalence between the entropy principle and the MEP was finally proved in 1997 by Boillat and Ruggeri \cite{Boillat-1997}. More details are found in \cite{Ruggeri-2021}. For non degenerate gases, the distribution function $f^{(N)}$ that maximizes the entropy \eqref{hgen} under the constraint that the first $N$ moments are prescribed is expressed by \cite{Muller-1993, Muller-1998, Ruggeri-2021}: 
\begin{equation}\label{v2}
    f^{(N)}= \exp{\left(-1-\frac{m}{k_B} \chi^{(N)}\right)},
\end{equation}
where $\chi^{(N)}$
%scalar product of the vector of Lagrange multipliers and the vector of the integrants (??????) of $N$ moments that 
is the generalization of Eq.~\eqref{ident} to the case with $N$ moments.

\smallskip 

Concerning the inequality \eqref{hmaxx}, since the function $H(f)$ defined in \eqref{HHH} is concave, we have the following inequality:
\begin{equation}\label{v1}
    H\left(f\right) \leq H\left(f_0\right) + H^\prime\left(f_0\right) \left(f-f_0\right). 
\end{equation}
Let us choose $f_0$ as the distribution function $f^{(N)}$, then, from Eq.~\eqref{v1}, Eq.~\eqref{v2} and Eq.~\eqref{hgen}, we have 
\begin{equation}\label{v3}
    h \leq h^{(N)} + m \inta \chi^{(N)} \left(f - f^{(N)}\right)\, \phiI \da.
\end{equation}
As the first $N$ moments of $f$ and $f^{(N)}$ are equal, the second term on the right-hand side of Eq.~\eqref{v3} disappears and then, the inequality \eqref{hmaxx} holds.

\section{Proof of Inequality (\ref{det_inequality})\label{App:det} }

As the proof closely follows the elegant method proposed by Dauvois et al. \cite{Dauvois-2021}, we provide a concise presentation. The primary distinction in our approach is the adoption of a single internal mode, unlike the original work.

Since $\det{\TT}$ is characterized by a parameter $\nu$, let us introduce
\begin{align}
	\phi\left(\nu\right) = \log\left(\det \TT\right) = \sum_{i=1}^3 \log \left\{\theta \frac{p}{\rho} + \left(1 - \theta\right)\frac{1}{\rho}\left( \nu \lambda_i^\pp + \left(1-\nu\right)\PP\right) \right\}.
\end{align}
This is a concave function because the argument of the logarithm function is positive due to the definite positiveness of $\TT$, and has a maximum at $\nu=0$ since $\phi'\left(0\right) = 0$. With the use of the arithmetic and geometric means, we can prove $\phi\left(-\frac{1}{2}\right) \geq \phi\left(1\right)$, and therefore $\phi\left(\nu\right) \geq \phi\left(1\right)$.
The derived inequality provides
\begin{align}
	\det \TT \geq \prod_{i=1}^3 \frac{1}{\rho} \left(\theta p + \left(1 - \theta\right)\lambda_i^\pp\right),
\end{align}
then, we have
\begin{align}
	\frac{\det\left(\pp/\rho \right)}{\det \TT}  \leq \frac{\prod_{i=1}^3 \lambda_i^\pp}{\prod_{i=1}^3 \left(\theta p + \left(1 - \theta\right)\lambda_i^\pp\right)}.
\end{align}
From this inequality, we obtain
\begin{align}\label{det_ineq2}
	\log \frac{\det\left(\pp/\rho \right)}{\det \TT}  &\leq \log \frac{\prod_{i=1}^3 \lambda_i^\pp}{\prod_{i=1}^3 \left(\theta p + \left(1 - \theta\right)\lambda_i^\pp\right)}
    = \sum_{i=1}^3 \log \left(\frac{\lambda_i^\pp}{\theta p + \left(1 - \theta\right)\lambda_i^\pp}\right)\\
    &\leq 3 \log \left(\frac{\PP}{\theta p + \left(1 - \theta\right)\PP}\right) 
    = \log \left(\frac{\TK}{\TKrel}\right)^3 = \log \left(\frac{\epsK_E\left(\TK\right)}{\epsK_E\left(\TKrel\right)}\right)^3.
\end{align}
Here we adopt $\sum_{i=1}^3 \lambda_i^\pp = 3\PP $ and utilize the Jensen inequality for a concave function in the second inequality. Then, the inequality Eq.~\eqref{det_inequality} is proven.

\end{appendices}

\backmatter

\bmhead{Acknowledgments}

The work of T.~Arima {was partially supported by} JST, PRESTO Grant Number JPMJPR23O1, Japan, and by JSPS KAKENHI Grant Numbers JP22K03912. This research was conducted during the sabbatical leave of author T.~Arima at the University of Bologna.

The work of A.~Mentrelli was partially supported by MUR under the PRIN2022 PNRR project n. P2022P5R22A, and by the Italian National Institute for Nuclear Physics (INFN), grant FLAG.

The work of A.~Mentrelli and T.~Ruggeri was partially supported and carried out in the framework of the activities of the Italian National Group for Mathematical Physics [Gruppo Nazionale per la Fisica Matematica (GNFM/INdAM)].

\bmhead{Data availability}
Not applicable.

\bmhead{Conflict of interest}
The authors have no conflicts of interest to declare.

\bibliography{sn-bibliography-AMR}


\begin{thebibliography}{46}
\ifx \bisbn   \undefined \def \bisbn  #1{ISBN #1}\fi
\ifx \binits  \undefined \def \binits#1{#1}\fi
\ifx \bauthor  \undefined \def \bauthor#1{#1}\fi
\ifx \batitle  \undefined \def \batitle#1{#1}\fi
\ifx \bjtitle  \undefined \def \bjtitle#1{#1}\fi
\ifx \bvolume  \undefined \def \bvolume#1{\textbf{#1}}\fi
\ifx \byear  \undefined \def \byear#1{#1}\fi
\ifx \bissue  \undefined \def \bissue#1{#1}\fi
\ifx \bfpage  \undefined \def \bfpage#1{#1}\fi
\ifx \blpage  \undefined \def \blpage #1{#1}\fi
\ifx \burl  \undefined \def \burl#1{\textsf{#1}}\fi
\ifx \doiurl  \undefined \def \doiurl#1{\url{https://doi.org/#1}}\fi
\ifx \betal  \undefined \def \betal{\textit{et al.}}\fi
\ifx \binstitute  \undefined \def \binstitute#1{#1}\fi
\ifx \binstitutionaled  \undefined \def \binstitutionaled#1{#1}\fi
\ifx \bctitle  \undefined \def \bctitle#1{#1}\fi
\ifx \beditor  \undefined \def \beditor#1{#1}\fi
\ifx \bpublisher  \undefined \def \bpublisher#1{#1}\fi
\ifx \bbtitle  \undefined \def \bbtitle#1{#1}\fi
\ifx \bedition  \undefined \def \bedition#1{#1}\fi
\ifx \bseriesno  \undefined \def \bseriesno#1{#1}\fi
\ifx \blocation  \undefined \def \blocation#1{#1}\fi
\ifx \bsertitle  \undefined \def \bsertitle#1{#1}\fi
\ifx \bsnm \undefined \def \bsnm#1{#1}\fi
\ifx \bsuffix \undefined \def \bsuffix#1{#1}\fi
\ifx \bparticle \undefined \def \bparticle#1{#1}\fi
\ifx \barticle \undefined \def \barticle#1{#1}\fi
\bibcommenthead
\ifx \bconfdate \undefined \def \bconfdate #1{#1}\fi
\ifx \botherref \undefined \def \botherref #1{#1}\fi
\ifx \url \undefined \def \url#1{\textsf{#1}}\fi
\ifx \bchapter \undefined \def \bchapter#1{#1}\fi
\ifx \bbook \undefined \def \bbook#1{#1}\fi
\ifx \bcomment \undefined \def \bcomment#1{#1}\fi
\ifx \oauthor \undefined \def \oauthor#1{#1}\fi
\ifx \citeauthoryear \undefined \def \citeauthoryear#1{#1}\fi
\ifx \endbibitem  \undefined \def \endbibitem {}\fi
\ifx \bconflocation  \undefined \def \bconflocation#1{#1}\fi
\ifx \arxivurl  \undefined \def \arxivurl#1{\textsf{#1}}\fi
\csname PreBibitemsHook\endcsname

%%% 1
\bibitem[\protect\citeauthoryear{McCourt et~al.}{1990}]{McCourt-1990}
\begin{bbook}
\bauthor{\bsnm{McCourt}, \binits{F.R.W.}},
\bauthor{\bsnm{Beenakker}, \binits{J.J.M.}},
\bauthor{\bsnm{K{\"o}hler}, \binits{W.E.}},
\bauthor{\bsnm{Ku\v{s}\v{c}er}, \binits{I.}}:
\bbtitle{{N}onequilibrium Phenomena in Polyatomic Gases, Vol. 1: {D}ilute Gases}.
\bpublisher{Clarendon Press},
\blocation{Oxford}
(\byear{1990})
\end{bbook}
\endbibitem

%%% 2
\bibitem[\protect\citeauthoryear{Nagnibeda and Kustova}{2009}]{Nagnibeda-2009}
\begin{bbook}
\bauthor{\bsnm{Nagnibeda}, \binits{E.}},
\bauthor{\bsnm{Kustova}, \binits{E.}}:
\bbtitle{{N}on-Equilibrium Reacting Gas Flows: {K}inetic Theory of Transport and Relaxation Processes}.
\bpublisher{Springer},
\blocation{Berlin}
(\byear{2009})
\end{bbook}
\endbibitem

%%% 3
\bibitem[\protect\citeauthoryear{Boyd and Schwartzentruber}{2017}]{Boyd-2017}
\begin{bbook}
\bauthor{\bsnm{Boyd}, \binits{I.D.}},
\bauthor{\bsnm{Schwartzentruber}, \binits{T.E.}}:
\bbtitle{{N}onequilibrium Gas Dynamics and Molecular Simulation}.
\bpublisher{Cambridge University Press},
\blocation{Cambridge}
(\byear{2017})
\end{bbook}
\endbibitem

%%% 4
\bibitem[\protect\citeauthoryear{Borsoni et~al.}{2022}]{Borsoni-2022}
\begin{barticle}
\bauthor{\bsnm{Borsoni}, \binits{T.}},
\bauthor{\bsnm{Bisi}, \binits{M.}},
\bauthor{\bsnm{Groppi}, \binits{M.}}:
\batitle{{A} general framework for the kinetic modelling of polyatomic gases}.
\bjtitle{Commun. {M}ath. {P}hys.}
\bvolume{393}(\bissue{1}),
\bfpage{215}--\blpage{266}
(\byear{2022})
\end{barticle}
\endbibitem

%%% 5
\bibitem[\protect\citeauthoryear{Li and Zhang}{2009}]{Li-2009}
\begin{barticle}
\bauthor{\bsnm{Li}, \binits{Z.-H.}},
\bauthor{\bsnm{Zhang}, \binits{H.-X.}}:
\batitle{{G}as-kinetic numerical studies of three-dimensional complex flows on spacecraft re-entry}.
\bjtitle{J. {C}omput. {P}hys.}
\bvolume{228}(\bissue{4}),
\bfpage{1116}--\blpage{1138}
(\byear{2009})
\end{barticle}
\endbibitem

%%% 6
\bibitem[\protect\citeauthoryear{Mathiaud}{2018}]{Mathiaud-2018}
\begin{bbook}
\bauthor{\bsnm{Mathiaud}, \binits{J.}}:
\bbtitle{{M}odels and Methods for Complex Flows: {A}pplication to Atmospheric Reentry and Particle/Fluid Interactions ({H}abilitation à Diriger des Recherches)}.
\bpublisher{University of Bordeaux},
\blocation{Bordeaux}
(\byear{2018})
\end{bbook}
\endbibitem

%%% 7
\bibitem[\protect\citeauthoryear{Borgnakke and Larsen}{1975}]{Borgnakke-1975}
\begin{barticle}
\bauthor{\bsnm{Borgnakke}, \binits{C.}},
\bauthor{\bsnm{Larsen}, \binits{P.S.}}:
\batitle{{S}tatistical collision model for {M}onte {C}arlo simulation of polyatomic gas mixture}.
\bjtitle{J. {C}omput. {P}hys.}
\bvolume{18}(\bissue{4}),
\bfpage{405}--\blpage{420}
(\byear{1975})
\end{barticle}
\endbibitem

%%% 8
\bibitem[\protect\citeauthoryear{Bourgat et~al.}{1994}]{Bourgat-1994}
\begin{barticle}
\bauthor{\bsnm{Bourgat}, \binits{J.-F.}},
\bauthor{\bsnm{Desvillettes}, \binits{L.}},
\bauthor{\bsnm{Le~Tallec}, \binits{P.}},
\bauthor{\bsnm{Perthame}, \binits{B.}}:
\batitle{{M}icroreversible collisions for polyatomic gases and {B}oltzmann's theorem}.
\bjtitle{Eur. {J}. {M}ech. B/{F}luids}
\bvolume{13}(\bissue{2}),
\bfpage{237}--\blpage{254}
(\byear{1994})
\end{barticle}
\endbibitem

%%% 9
\bibitem[\protect\citeauthoryear{Kosuge et~al.}{2019}]{Kosuge-2019}
\begin{barticle}
\bauthor{\bsnm{Kosuge}, \binits{S.}},
\bauthor{\bsnm{Kuo}, \binits{H.-W.}},
\bauthor{\bsnm{Aoki}, \binits{K.}}:
\batitle{{A} kinetic model for a polyatomic gas with temperature-dependent specific heats and its application to shock-wave structure}.
\bjtitle{J. {S}tat. {P}hys.}
\bvolume{177}(\bissue{2}),
\bfpage{209}--\blpage{251}
(\byear{2019})
\end{barticle}
\endbibitem

%%% 10
\bibitem[\protect\citeauthoryear{Ruggeri and Sugiyama}{2021}]{Ruggeri-2021}
\begin{bbook}
\bauthor{\bsnm{Ruggeri}, \binits{T.}},
\bauthor{\bsnm{Sugiyama}, \binits{M.}}:
\bbtitle{{C}lassical and Relativistic Rational Extended Thermodynamics of Gases}.
\bpublisher{Springer},
\blocation{Cham}
(\byear{2021})
\end{bbook}
\endbibitem

%%% 11
\bibitem[\protect\citeauthoryear{Pavi\'c et~al.}{2013}]{Pavic-2013}
\begin{barticle}
\bauthor{\bsnm{Pavi\'c}, \binits{M.}},
\bauthor{\bsnm{Ruggeri}, \binits{T.}},
\bauthor{\bsnm{Simi\'c}, \binits{S.}}:
\batitle{{M}aximum entropy principle for rarefied polyatomic gases}.
\bjtitle{Phys. A: {S}tat. {M}ech. {A}ppl.}
\bvolume{392}(\bissue{6}),
\bfpage{1302}--\blpage{1317}
(\byear{2013})
\end{barticle}
\endbibitem

%%% 12
\bibitem[\protect\citeauthoryear{Ruggeri}{2021}]{Ruggeri-2021-MEP}
\begin{barticle}
\bauthor{\bsnm{Ruggeri}, \binits{T.}}:
\batitle{{M}aximum entropy principle closure for 14-moment system for a non-polytropic gas}.
\bjtitle{Ricerche {M}at.}
\bvolume{70}(\bissue{1}),
\bfpage{207}--\blpage{222}
(\byear{2021})
\end{barticle}
\endbibitem

%%% 13
\bibitem[\protect\citeauthoryear{Kogan}{1967}]{Kogan-1969}
\begin{bbook}
\bauthor{\bsnm{Kogan}, \binits{M.N.}}:
\bbtitle{Rarefied Gas Dynamics}
vol. \bseriesno{I},
pp. \bfpage{359}--\blpage{368}.
\bpublisher{Academic Press},
\blocation{New York}
(\byear{1967})
\end{bbook}
\endbibitem

%%% 14
\bibitem[\protect\citeauthoryear{Dreyer}{1987}]{Dreyer-1987}
\begin{barticle}
\bauthor{\bsnm{Dreyer}, \binits{W.}}:
\batitle{{M}aximisation of the entropy in non-equilibrium}.
\bjtitle{J. {P}hys. A: {M}ath. {G}en.}
\bvolume{20}(\bissue{18}),
\bfpage{6505}--\blpage{6517}
(\byear{1987})
\end{barticle}
\endbibitem

%%% 15
\bibitem[\protect\citeauthoryear{M{\"u}ller and Ruggeri}{1993}]{Muller-1993}
\begin{bbook}
\bauthor{\bsnm{M{\"u}ller}, \binits{I.}},
\bauthor{\bsnm{Ruggeri}, \binits{T.}}:
\bbtitle{Extended Thermodynamics}.
\bpublisher{Springer},
\blocation{New York}
(\byear{1993})
\end{bbook}
\endbibitem

%%% 16
\bibitem[\protect\citeauthoryear{Arima et~al.}{2012}]{Arima-2012}
\begin{barticle}
\bauthor{\bsnm{Arima}, \binits{T.}},
\bauthor{\bsnm{Taniguchi}, \binits{S.}},
\bauthor{\bsnm{Ruggeri}, \binits{T.}},
\bauthor{\bsnm{Sugiyama}, \binits{M.}}:
\batitle{{E}xtended thermodynamics of dense gases}.
\bjtitle{Continuum {M}ech. {T}hermodyn.}
\bvolume{24}(\bissue{4--6}),
\bfpage{271}--\blpage{292}
(\byear{2012})
\end{barticle}
\endbibitem

%%% 17
\bibitem[\protect\citeauthoryear{Arima et~al.}{2016}]{Arima-2016}
\begin{barticle}
\bauthor{\bsnm{Arima}, \binits{T.}},
\bauthor{\bsnm{Taniguchi}, \binits{S.}},
\bauthor{\bsnm{Ruggeri}, \binits{T.}},
\bauthor{\bsnm{Sugiyama}, \binits{M.}}:
\batitle{{R}ecent resuls on non linear extended thermodynamics of real gas with six fields, {P}art {I}: {G}eneral theory}.
\bjtitle{Ricerche {M}at.}
\bvolume{65}(\bissue{1}),
\bfpage{263}--\blpage{277}
(\byear{2016})
\end{barticle}
\endbibitem

%%% 18
\bibitem[\protect\citeauthoryear{Bisi et~al.}{2018}]{Bisi-2018}
\begin{barticle}
\bauthor{\bsnm{Bisi}, \binits{M.}},
\bauthor{\bsnm{Ruggeri}, \binits{T.}},
\bauthor{\bsnm{Spiga}, \binits{G.}}:
\batitle{{D}ynamical pressure in a polyatomic gas: {I}nterplay between kinetic theory and extended thermodynamics}.
\bjtitle{Kinet. {R}el. {M}odels}
\bvolume{11}(\bissue{2}),
\bfpage{1}--\blpage{25}
(\byear{2018})
\end{barticle}
\endbibitem

%%% 19
\bibitem[\protect\citeauthoryear{Arima et~al.}{2021}]{Arima-2021}
\begin{barticle}
\bauthor{\bsnm{Arima}, \binits{T.}},
\bauthor{\bsnm{Carrisi}, \binits{M.C.}},
\bauthor{\bsnm{Pennisi}, \binits{S.}},
\bauthor{\bsnm{Ruggeri}, \binits{T.}}:
\batitle{{W}hich moments are appropriate to describe gases with internal structure in rational extended thermodynamics?}
\bjtitle{Int. {J}. {N}on {L}inear {M}ech.}
\bvolume{137},
\bfpage{103820}
(\byear{2021})
\end{barticle}
\endbibitem

%%% 20
\bibitem[\protect\citeauthoryear{Baranger et~al.}{2020}]{Baranger-2020}
\begin{barticle}
\bauthor{\bsnm{Baranger}, \binits{C.}},
\bauthor{\bsnm{Dauvois}, \binits{Y.}},
\bauthor{\bsnm{Marois}, \binits{G.}},
\bauthor{\bsnm{Mathé}, \binits{J.}},
\bauthor{\bsnm{Mathiaud}, \binits{J.}},
\bauthor{\bsnm{Mieussens}, \binits{L.}}:
\batitle{{A} {BGK} model for high temperature rarefied gas flows}.
\bjtitle{Eur. {J}. {M}ech. B/{F}luids}
\bvolume{80},
\bfpage{1}--\blpage{12}
(\byear{2020})
\end{barticle}
\endbibitem

%%% 21
\bibitem[\protect\citeauthoryear{Bisi and Spiga}{2017}]{Bisi-2017}
\begin{barticle}
\bauthor{\bsnm{Bisi}, \binits{M.}},
\bauthor{\bsnm{Spiga}, \binits{G.}}:
\batitle{{O}n kinetic models for polyatomic gases and their hydrodynamic limits}.
\bjtitle{Ricerche {M}at.}
\bvolume{66}(\bissue{1}),
\bfpage{113}--\blpage{124}
(\byear{2017})
\end{barticle}
\endbibitem

%%% 22
\bibitem[\protect\citeauthoryear{Rahimi and Struchtrup}{2014}]{Rahimi-2014}
\begin{barticle}
\bauthor{\bsnm{Rahimi}, \binits{B.}},
\bauthor{\bsnm{Struchtrup}, \binits{H.}}:
\batitle{{C}apturing non-equilibrium phenomena in rarefied polyatomic gases: {A} high-order macroscopic model}.
\bjtitle{Phys. {F}luids}
\bvolume{26}(\bissue{5}),
\bfpage{052001}
(\byear{2014})
\end{barticle}
\endbibitem

%%% 23
\bibitem[\protect\citeauthoryear{Struchtrup}{1999}]{Struchtrup-1999}
\begin{barticle}
\bauthor{\bsnm{Struchtrup}, \binits{H.}}:
\batitle{{T}he {BGK} model for an ideal gas with an internal degree of freedom}.
\bjtitle{Transp. {T}heor. {S}tat. {P}hys.}
\bvolume{28}(\bissue{4}),
\bfpage{369}--\blpage{385}
(\byear{1999})
\end{barticle}
\endbibitem

%%% 24
\bibitem[\protect\citeauthoryear{Holway}{1966}]{Holway-1966}
\begin{barticle}
\bauthor{\bsnm{Holway}, \binits{L.H.J.}}:
\batitle{{N}ew statistical models for kinetic theory: {M}ethods of construction}.
\bjtitle{Phys. {F}luids}
\bvolume{9}(\bissue{9}),
\bfpage{1658}--\blpage{1673}
(\byear{1966})
\end{barticle}
\endbibitem

%%% 25
\bibitem[\protect\citeauthoryear{Andries et~al.}{2000}]{Andries-2000}
\begin{barticle}
\bauthor{\bsnm{Andries}, \binits{P.}},
\bauthor{\bsnm{Le~Tallec}, \binits{P.}},
\bauthor{\bsnm{Perlat}, \binits{J.-P.}},
\bauthor{\bsnm{Perthame}, \binits{B.}}:
\batitle{{T}he {G}aussian-{BGK} model of {B}oltzmann equation with small {P}randtl number}.
\bjtitle{Eur. {J}. {M}ech. B/{F}luids}
\bvolume{19}(\bissue{6}),
\bfpage{813}--\blpage{830}
(\byear{2000})
\end{barticle}
\endbibitem

%%% 26
\bibitem[\protect\citeauthoryear{Brull and Schneider}{2009}]{Brull-2009}
\begin{barticle}
\bauthor{\bsnm{Brull}, \binits{S.}},
\bauthor{\bsnm{Schneider}, \binits{J.}}:
\batitle{{O}n the ellipsoidal statistical model for polyatomic gases}.
\bjtitle{Continuum {M}ech. {T}hermodyn.}
\bvolume{20}(\bissue{8}),
\bfpage{489}--\blpage{508}
(\byear{2009})
\end{barticle}
\endbibitem

%%% 27
\bibitem[\protect\citeauthoryear{Brau}{1967}]{Brau-1967}
\begin{barticle}
\bauthor{\bsnm{Brau}, \binits{C.A.}}:
\batitle{{K}inetic theory of polyatomic gases: {M}odels for the collision processes}.
\bjtitle{Phys. {F}luids}
\bvolume{10}(\bissue{1}),
\bfpage{48}--\blpage{55}
(\byear{1967})
\end{barticle}
\endbibitem

%%% 28
\bibitem[\protect\citeauthoryear{Gorji and Jenny}{2013}]{Gorji-2013}
\begin{barticle}
\bauthor{\bsnm{Gorji}, \binits{M.H.}},
\bauthor{\bsnm{Jenny}, \binits{P.}}:
\batitle{{A} {F}okker–{P}lanck based kinetic model for diatomic rarefied gas flows}.
\bjtitle{Phys. {F}luids}
\bvolume{25}(\bissue{6}),
\bfpage{062002}
(\byear{2013})
\end{barticle}
\endbibitem

%%% 29
\bibitem[\protect\citeauthoryear{Mathiaud and Mieussens}{2017}]{Mathiaud-2017}
\begin{barticle}
\bauthor{\bsnm{Mathiaud}, \binits{J.}},
\bauthor{\bsnm{Mieussens}, \binits{L.}}:
\batitle{{A} {F}okker–{P}lanck model of the {B}oltzmann equation with correct {P}randtl number for polyatomic gases}.
\bjtitle{J. {S}tat. {P}hys.}
\bvolume{168}(\bissue{5}),
\bfpage{1031}--\blpage{1055}
(\byear{2017})
\end{barticle}
\endbibitem

%%% 30
\bibitem[\protect\citeauthoryear{Arima et~al.}{2017}]{Arima-2017}
\begin{barticle}
\bauthor{\bsnm{Arima}, \binits{T.}},
\bauthor{\bsnm{Ruggeri}, \binits{T.}},
\bauthor{\bsnm{Sugiyama}, \binits{M.}}:
\batitle{{R}ational extended thermodynamics of a rarefied polyatomic gas with molecular relaxation processes}.
\bjtitle{Phys. {R}ev. E}
\bvolume{96}(\bissue{4}),
\bfpage{042143}
(\byear{2017})
\end{barticle}
\endbibitem

%%% 31
\bibitem[\protect\citeauthoryear{Arima et~al.}{2018}]{Arima-2018}
\begin{barticle}
\bauthor{\bsnm{Arima}, \binits{T.}},
\bauthor{\bsnm{Ruggeri}, \binits{T.}},
\bauthor{\bsnm{Sugiyama}, \binits{M.}}:
\batitle{{E}xtended thermodynamics of rarefied polyatomic gases: 15-field theory incorporating relaxation processes of molecular rotation and vibration}.
\bjtitle{Entropy}
\bvolume{20}(\bissue{4}),
\bfpage{301}
(\byear{2018})
\end{barticle}
\endbibitem

%%% 32
\bibitem[\protect\citeauthoryear{Dauvois et~al.}{2021}]{Dauvois-2021}
\begin{barticle}
\bauthor{\bsnm{Dauvois}, \binits{Y.}},
\bauthor{\bsnm{Mathiaud}, \binits{J.}},
\bauthor{\bsnm{Mieussens}, \binits{L.}}:
\batitle{{A}n {ES}-{BGK} model for polyatomic gases in rotational and vibrational nonequilibrium}.
\bjtitle{Eur. {J}. {M}ech. B/{F}luids}
\bvolume{88},
\bfpage{1}--\blpage{16}
(\byear{2021})
\end{barticle}
\endbibitem

%%% 33
\bibitem[\protect\citeauthoryear{Mathiaud et~al.}{2022}]{Mathiaud-2022}
\begin{barticle}
\bauthor{\bsnm{Mathiaud}, \binits{J.}},
\bauthor{\bsnm{Mieussens}, \binits{L.}},
\bauthor{\bsnm{Pfeiffer}, \binits{M.}}:
\batitle{{A}n {ES}-{BGK} model for diatomic gases with correct relaxation rates for internal energies}.
\bjtitle{Eur. {J}. {M}ech. B/{F}luids}
\bvolume{96},
\bfpage{65}--\blpage{77}
(\byear{2022})
\end{barticle}
\endbibitem

%%% 34
\bibitem[\protect\citeauthoryear{M{\"u}ller and Ruggeri}{1998}]{Muller-1998}
\begin{bbook}
\bauthor{\bsnm{M{\"u}ller}, \binits{I.}},
\bauthor{\bsnm{Ruggeri}, \binits{T.}}:
\bbtitle{{R}ational Extended Thermodynamics (2nd Edition)}.
\bpublisher{Springer},
\blocation{New York}
(\byear{1998})
\end{bbook}
\endbibitem

%%% 35
\bibitem[\protect\citeauthoryear{Boillat and Ruggeri}{1997}]{Boillat-1997}
\begin{barticle}
\bauthor{\bsnm{Boillat}, \binits{G.}},
\bauthor{\bsnm{Ruggeri}, \binits{T.}}:
\batitle{{M}oment equations in the kinetic theory of gases and wave velocities}.
\bjtitle{Continuum {M}ech. {T}hermodyn.}
\bvolume{9}(\bissue{4}),
\bfpage{205}--\blpage{212}
(\byear{1997})
\end{barticle}
\endbibitem

%%% 36
\bibitem[\protect\citeauthoryear{Ruggeri}{1989}]{Ruggeri-1989}
\begin{barticle}
\bauthor{\bsnm{Ruggeri}, \binits{T.}}:
\batitle{{G}alilean invariance and entropy principle for systems of balance laws}.
\bjtitle{Continuum {M}ech. {T}hermodyn.}
\bvolume{1}(\bissue{1}),
\bfpage{3}--\blpage{20}
(\byear{1989})
\end{barticle}
\endbibitem

%%% 37
\bibitem[\protect\citeauthoryear{Hamburger}{1944}]{Hamburger-1944}
\begin{barticle}
\bauthor{\bsnm{Hamburger}, \binits{H.L.}}:
\batitle{Hermitian transformations of deficiency-index (1, 1), {J}acobi matrices and undetermined moment problems}.
\bjtitle{Am. {J}. {M}ath.}
\bvolume{66}(\bissue{4}),
\bfpage{489}--\blpage{522}
(\byear{1944})
\end{barticle}
\endbibitem

%%% 38
\bibitem[\protect\citeauthoryear{Arima et~al.}{2012}]{Arima-2012-RET6}
\begin{barticle}
\bauthor{\bsnm{Arima}, \binits{T.}},
\bauthor{\bsnm{Taniguchi}, \binits{S.}},
\bauthor{\bsnm{Ruggeri}, \binits{T.}},
\bauthor{\bsnm{Sugiyama}, \binits{M.}}:
\batitle{{E}xtended thermodynamics of real gases with dynamic pressure: {A}n extension of {M}eixner's theory}.
\bjtitle{Phys. {L}ett. A}
\bvolume{376}(\bissue{44}),
\bfpage{2799}--\blpage{2803}
(\byear{2012})
\end{barticle}
\endbibitem

%%% 39
\bibitem[\protect\citeauthoryear{Arima et~al.}{2013}]{Arima-2013-LW}
\begin{barticle}
\bauthor{\bsnm{Arima}, \binits{T.}},
\bauthor{\bsnm{Taniguchi}, \binits{S.}},
\bauthor{\bsnm{Ruggeri}, \binits{T.}},
\bauthor{\bsnm{Sugiyama}, \binits{M.}}:
\batitle{{D}ispersion relation for sound in rarefied polyatomic gases based on extended thermodynamics}.
\bjtitle{Continuum {M}ech. {T}hermodyn.}
\bvolume{25}(\bissue{6}),
\bfpage{727}--\blpage{737}
(\byear{2013})
\end{barticle}
\endbibitem

%%% 40
\bibitem[\protect\citeauthoryear{Taniguchi et~al.}{2014}]{Taniguchi-2014}
\begin{barticle}
\bauthor{\bsnm{Taniguchi}, \binits{S.}},
\bauthor{\bsnm{Arima}, \binits{T.}},
\bauthor{\bsnm{Ruggeri}, \binits{T.}},
\bauthor{\bsnm{Sugiyama}, \binits{M.}}:
\batitle{{T}hermodynamic theory of the shock wave structure in a rarefied polyatomic gas: {B}eyond the {B}ethe-{T}eller theory}.
\bjtitle{Phys. {R}ev. E}
\bvolume{89}(\bissue{1}),
\bfpage{013025}
(\byear{2014})
\end{barticle}
\endbibitem

%%% 41
\bibitem[\protect\citeauthoryear{Gilbarg and Paolucci}{1953}]{Gilbarg-1953}
\begin{barticle}
\bauthor{\bsnm{Gilbarg}, \binits{G.}},
\bauthor{\bsnm{Paolucci}, \binits{D.}}:
\batitle{{T}he structure of shock waves in the continuum theory of fluids}.
\bjtitle{J. {R}ational {M}ech.}
\bvolume{2}(\bissue{4}),
\bfpage{617}--\blpage{642}
(\byear{1953})
\end{barticle}
\endbibitem

%%% 42
\bibitem[\protect\citeauthoryear{Kosuge and Aoki}{2018}]{Kosuge-2018}
\begin{barticle}
\bauthor{\bsnm{Kosuge}, \binits{S.}},
\bauthor{\bsnm{Aoki}, \binits{K.}}:
\batitle{{S}hock-wave structure for a polyatomic gas with large bulk viscosity}.
\bjtitle{Phys. {R}ev. {F}luids}
\bvolume{3}(\bissue{2}),
\bfpage{023401}
(\byear{2018})
\end{barticle}
\endbibitem

%%% 43
\bibitem[\protect\citeauthoryear{Taniguchi et~al.}{2016}]{Taniguchi-2016}
\begin{barticle}
\bauthor{\bsnm{Taniguchi}, \binits{S.}},
\bauthor{\bsnm{Arima}, \binits{T.}},
\bauthor{\bsnm{Ruggeri}, \binits{T.}},
\bauthor{\bsnm{Sugiyama}, \binits{M.}}:
\batitle{{O}vershoot of the non-equilibrium temperature in the shock wave structure of a rareﬁed polyatomic gas subject to the dynamic pressure}.
\bjtitle{Int. {J}. {N}on {L}inear {M}ech.}
\bvolume{79},
\bfpage{66}--\blpage{75}
(\byear{2016})
\end{barticle}
\endbibitem

%%% 44
\bibitem[\protect\citeauthoryear{Jaynes}{1957a}]{Jaynes-1957a}
\begin{barticle}
\bauthor{\bsnm{Jaynes}, \binits{E.T.}}:
\batitle{{I}nformation theory and statistical mechanics}.
\bjtitle{{P}hys. {R}ev.}
\bvolume{106}(\bissue{4}),
\bfpage{620}--\blpage{630}
(\byear{1957})
\end{barticle}
\endbibitem

%%% 45
\bibitem[\protect\citeauthoryear{Jaynes}{1957b}]{Jaynes-1957b}
\begin{barticle}
\bauthor{\bsnm{Jaynes}, \binits{E.T.}}:
\batitle{{I}nformation theory and statistical mechanics {II}}.
\bjtitle{{P}hys. {R}ev.}
\bvolume{108}(\bissue{2}),
\bfpage{171}--\blpage{190}
(\byear{1957})
\end{barticle}
\endbibitem

%%% 46
\bibitem[\protect\citeauthoryear{Levermore}{1996}]{Levermore-1996}
\begin{barticle}
\bauthor{\bsnm{Levermore}, \binits{C.D.}}:
\batitle{{M}oment closure hierarchies for kinetic theories}.
\bjtitle{J. {S}tat. {P}hys.}
\bvolume{83}(\bissue{5/6}),
\bfpage{1021}--\blpage{1065}
(\byear{1996})
\end{barticle}
\endbibitem

\end{thebibliography}

\end{document}